\documentclass[conference]{IEEEtran}
\IEEEoverridecommandlockouts
\usepackage{cite}
\usepackage{amsmath,amsfonts,amsthm,amssymb, bm}
\usepackage{algorithmic}
\usepackage{enumitem, graphicx}
\usepackage{textcomp}
\usepackage{xcolor}
\def\BibTeX{{\rm B\kern-.05em{\sc i\kern-.025em b}\kern-.08em
    T\kern-.1667em\lower.7ex\hbox{E}\kern-.125emX}}

\usepackage[section]{placeins}
\usepackage{tikz}
\usetikzlibrary{shapes,arrows,positioning,calc}
\theoremstyle{remark}
\newtheorem{remark}{Remark}

\newtheorem{theorem}{Theorem}
\newtheorem{corollary}{Corollary}[theorem]

\usepackage[linesnumbered,ruled,vlined]{algorithm2e}

\begin{document}
\title{Novel Active Sensing and Inference for mmWave Beam Alignment Using Single RF Chain Systems
\thanks{This research was supported by National Science Foundation (NSF) under Grant CCF-2124929 and Grant CCF-2225617.}
\thanks{The authors are with the Department of Electrical and Computer Engineering, University of California, San Diego, La Jolla, CA 92092, USA (rpote@ucsd.edu, brao@ucsd.edu).}
}
\author{\IEEEauthorblockN{Rohan R. Pote}
\IEEEauthorblockA{\textit{ECE Department, UC San Diego}}
\and
\IEEEauthorblockN{Bhaskar D. Rao}
\IEEEauthorblockA{\textit{ECE Department, UC San Diego} 
}}
\maketitle
\begin{abstract}
We propose a novel sensing approach for the beam alignment problem in millimeter wave systems using a single Radio Frequency (RF) chain. Conventionally, beam alignment using a single phased array involves comparing beamformer output power
across different spatial regions. This incurs large training overhead due to the need to perform the beam scan operation. The proposed Synthesis of Virtual Array Manifold (SVAM) sensing methodology is inspired from synthetic aperture radar systems and realizes a virtual array geometry over temporal measurements.
We demonstrate the benefits of SVAM using Cram\'er-Rao bound (CRB) analysis over schemes that repeat beam pattern to boost signal-to-noise (SNR) ratio. We also showcase versatile applicability of the proposed SVAM sensing by incorporating it within existing beam alignment procedures that assume perfect knowledge of the small-scale fading coefficient.
We further consider the practical scenario wherein we estimate the fading coefficient and propose a novel beam alignment procedure based on efficient computation of an approximate posterior density on dominant path angle. We provide numerical experiments to study the impact of parameters involved in the procedure. The performance of the proposed sensing and beam alignment algorithm is empirically observed to approach the fading coefficient-perfectly known performance, even at low SNR.

\end{abstract}
\newlength{\textfloatsepsave} \setlength{\textfloatsepsave}{\textfloatsep} \setlength{\textfloatsep}{0pt}
\section{Introduction}
Millimeter wave (mmWave) technology is essential for expanding the existing capabilities of cellular networks\cite{pi11, rappaport13}. The availability of the large spectrum in the 30-300 GHz spectrum range, and the ability to place many more antennas in the same form factor on a device are very promising avenues for the next generation wireless systems. This has propelled the interest, both in the industry and academia. The envisioned benefits include much higher throughput and low latency. The impact of this technology can be gauged from the numerous use cases enabled by the mmWave technology, which includes industrial-IoT, virtual/augmented reality, biomedical applications, and non-terrestrial networks\cite{bastug17,ghosh19}.   

MmWave technology also faces many challenges. 
The mmWave channel incurs large propagation losses thereby restricting coverage per base station (BS), and requiring additional infrastructure compared to legacy cellular networks. A second challenge is the \emph{sparse} nature of mmWave channel consequently requiring accurate beam alignment. 
This challenge is only further exacerbated by the narrow beamwidths and consequent large codebook size due to the large antenna array dimensions.
Several options are being considered for enhancing coverage at low-cost such as integrated access backhaul and intelligent reflective surfaces. On the other hand, reducing the beam alignment phase duration is a critical and active area of research. Hardware cost also impacts the ability of the transceivers to sense the mmWave channel. The large number of antenna elements are typically supported by only a few Radio Frequency (RF) chains\cite{doan04}, and thus necessitates for a low-dimensional projection of the received signal at the antennas. Beam alignment using such a low-dimensional signal is a challenging problem, and it is the main focus of this work.


There is a need to build a better sensing approach coupled with efficient inference mechanisms that exploit the array geometry and channel characteristics under the hardware constraints. An interesting direction adopted in \cite{chiu19} formulates the beam alignment problem under the posterior matching framework, and actively learns the single path Direction of Arrival (DoA). The work is shown to improve over the detection-based algorithm in \cite{alkhateeb14}. However, the authors in \cite{chiu19} assume that the small-scale fading coefficient is perfectly known. Subsequent effort build upon this work, and estimate both the DoA as well as the fading coefficient \cite{ronquillo19}. The Kalman filter-based posterior matching algorithm in \cite{ronquillo19} still requires good prior density on the small-scale fading coefficient. Similar assumptions on availability of good prior density was made in the variational hierarchical posterior matching algorithm proposed in \cite{akdim20}.

In this work, we reflect on the beam alignment problem from the perspective of active sensing for improved estimation performance. We do so
without relying on additional information such as good prior knowledge about the small-scale fading coefficient.  
The contributions of this work are as follows:\begin{itemize}[leftmargin=*]
    \item A novel sensing methodology, inspired from Synthetic Aperture Radar (SAR), is proposed for the single RF chain mmWave systems. Under the proposed sensing approach, a virtual Uniform Linear Array (ULA) manifold is synthesized over temporal measurements. Extension to construct a virtual arbitrary array geometry such as Sparse Linear Array (SLA) is also discussed. The proposed sensing is useful even in the presence of multipaths and its operation under such scenarios is briefly discussed.
    \item Benefits of the proposed sensing are described, when the channel small-scale fading coefficient is known, in terms of i. its impact on the Cram\'er-Rao lower Bound (CRB) on the variance of estimation of the unknown dominant path angle, when compared to a benchmark scheme, and ii. its ability to be incorporated within existing active beam alignment procedures. The improvement in terms of lower training overhead is demonstrated using numerical experiments.
    \item A novel beam alignment procedure is proposed which adapts the beamformer based on the current estimate of the posterior on the unknown angle. The proposed algorithm estimates a posterior on the small-scale fading along with the angular posterior. Both, flexible and hierarchical codebook-based beam alignment procedures are presented.
    \item Finally, the proposed sensing and beam alignment procedures are empirically studied, and compared with the performance using perfect knowledge of the channel state information. Impact of the parameters involved in the estimation procedure is also studied, which also reveal the ability of the adaptive beam alignment procedure to self-correct in case of premature misalignment during the early phase of the training period.   
\end{itemize}
\subsection{Relevant Prior Work}
A virtual array synthesis from spatial measurements under the reduced number of RF chains constraint for mmWave systems was proposed in \cite{pote19}.
A similar sensing scheme as in \cite{pote19}, was proposed in \cite{lin20} for mmWave multipath angle estimation and in \cite{chen20} for the DoA estimation problem. In \cite{lin20}, the authors proposed using random precoders and combiners that are submatrices of banded Toeplitz matrices. The work focuses on mmWave systems with multiple RF chains, and the ideas are extended to the single RF chain case. In \cite{chen20}, the authors investigated the applicability of root-MUSIC and ESPRIT algorithms, as a consequence of preserving the Vandermonde structure and the shift-invariance under the virtual array synthesis procedure. These sensing methodologies can be traced back to Silverstein\cite{silverstein91} and Tkacenko \cite{tkacenko01}. The sensing scheme proposed in this work synthesizes a virtual array over \emph{temporal measurements}, and considers the practical single phased array system. To the best of our knowledge, the presented adaptive sensing methodology is the first of its kind. 

Many non-adaptive and adaptive beamforming approaches have been proposed for the mmWave beam alignment problem in the past. Random beamforming was proposed in \cite{ramasamy12,berraki14,alkhateeb15}, wherein the inference was carried out using compressed sensing algorithms. This approach does not exploit beamforming gain needed to combat large path loss in the mmWave channels. An exhaustive beam steering approach proposed in IEEE 802.11ad and 5G standards improves the beamforming gain, but can be slow in selecting the appropriate beam. The acquisition time is improved in literature by replacing the \emph{non-adaptive} linear search with an \emph{adaptive} binary (or in general $n$-ary, $n\geq 2$) search within a hierarchical codebook \cite{hur13,alkhateeb14,xiao16,zhang17}. A comparison in terms of asymptotic misalignment probability between the exhaustive search and the hierarchical search was studied in \cite{liu17}. For adaptive schemes using hierarchical codebook, the inference at each hierarchical level is typically carried out by comparing power at the output of different beams within a hierarchical node. The inference mechanism was improved in \cite{chiu19} by computing posterior on the dominant path angle and selecting next beam based on the posterior within the hierarchical codebook of \cite{alkhateeb14}. The work assumes that the small-scale fading coefficient is perfectly known. However, such assumption is difficult to satisfy in practice. The problem of estimating small-scale fading coefficient (along with the unknown path angle) was considered in later works \cite{ronquillo19,akdim20}.
A grid-approach was also proposed in \cite{ronquillo19}, but the relevant issue on how to choose appropriate grid on the small-scale fading coefficient was not addressed. An adaptive beam search algorithm scheme based on posterior computation to compare beams was proposed in \cite{liu22}. Many learning-based approaches have been proposed in literature as well\cite{sohrabi22,wei23,hussain19}. These include approaches that frame the beam alignment problem as multi-armed bandit problem,
or train an end-to-end neural network to design a model-free or codebook-free architecture. In this work, we emphasize the model and propose novel sensing (for improved acquisition) and inference procedures to estimate the unknown parameters. The paper builds on our previous work in \cite{pote23asil}.
\subsection{Organization of the Paper and Notations}
We describe the problem tackled in this work, and introduce the new sensing approach in Section~\ref{sec:probstate_propsens}. In Section~\ref{sec:alphaknown}, we investigate the impact of the proposed sensing for the case when the fading coefficient is perfectly known, and only the unknown DoA is to be recovered. The more practical case, when the fading coefficient needs to be estimated along with the DoA is discussed in Section~\ref{sec:alphaunknown}. We provide empirical results in Section~\ref{sec:numsecalphaunknown}, and present our conclusions in Section~\ref{sec:conc}. 

{\it Notations:} We represent scalars, vectors, and matrices by lowercase, boldface-lowercase, and boldface-uppercase letters, respectively. Sets are represented using blackboard bold letters. $(.)^T,(.)^H,(.)^c$ denotes transpose, Hermitian, and complex conjugate operation respectively. $\otimes$ denotes matrix Kronecker product, and $\odot$ denotes Hadamard product of two conformable matrices. $*$ denotes convolution operation. $[M]=\{0,1,\ldots,M-1\},M\in\mathbb{Z}^+$. 
\section{Problem Statement \& Proposed Novel Sensing}\label{sec:probstate_propsens}
We consider a receiver (base station or user equipment) equipped with a Uniform Linear antenna Array (ULA) of size $N$
and a single RF chain. We assume a flat fading channel, with a single dominant path between the transmitter and receiver. We further assume that the channel remains coherent within the training duration due to low receiver mobility.
\subsection{Problem Statement}
\begin{figure*}
    \centering
    \includegraphics[width=0.8\linewidth]{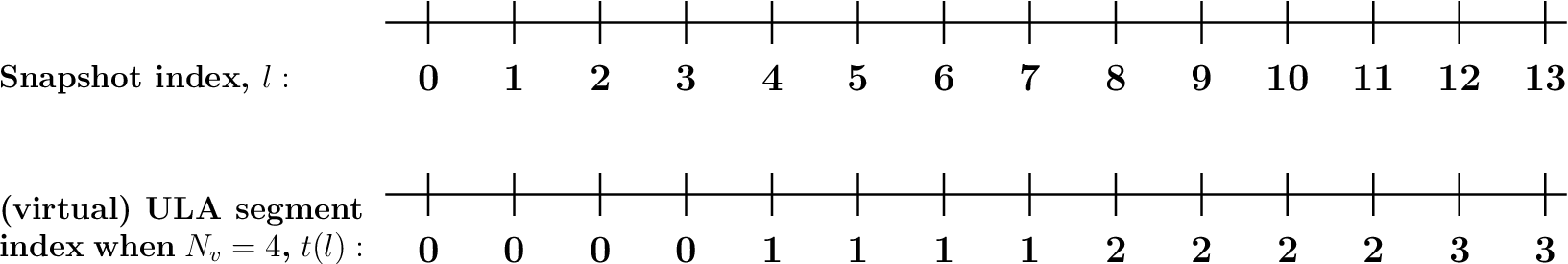}
    \caption{A (virtual) ULA segment of size $N_v=4$ is created using $4$ snapshots. The beamformer, $\mathbf{f}_{t(l)}$, is adapted once in every segment duration.}\vspace{-1em}
    \label{fig:numerology}
\end{figure*}
The received signal at the antennas at instant $l$, $\mathbf{x}_l\in\mathbb{C}^N$, is given by \begin{equation}
    \mathbf{x}_l=\sqrt{P_s}\alpha\bm{\phi}_N(u)+\bar{\mathbf{n}}_l,\quad l\in [L],\label{eq:aesignal}
\end{equation}where $L$ denotes the total training duration. $P_s\geq 0$ denotes the combined contribution of transmitted power and the large-scale fading (path loss and shadowing), $\alpha\in\mathbb{C}$ is the \emph{unknown} small-scale fading coefficient. Since the transmitted symbol is known, it can be easily absorbed within the signal term of $\mathbf{x}_l$ in (\ref{eq:aesignal}). Thus, we assume the transmitted symbol value to be $1$ without loss of generality. $\bm{\phi}_N(u)$ is the array manifold or response vector for an incoming narrowband signal along the angle $u$; $u=\sin\theta,u\in[-1,1),$ where $\theta\in[-\frac{\pi}{2},\frac{\pi}{2})$ denotes physical DoA. Noise, $\bar{\mathbf{n}}_l\in\mathbb{C}^N$, is distributed as $\mathcal{CN}(\mathbf{0},\sigma_n^2\mathbf{I})$ and i.i.d. over time. 
Since there is a single RF chain, the received signal is processed using an analog combiner, $\mathbf{w}_l\in\mathbb{C}^N$. The output, $y_l\in\mathbb{C}$, available for inference is given by\begin{IEEEeqnarray}{ll}   y_l=\mathbf{w}_l^H\mathbf{x}_l&=\sqrt{P_s}\alpha\mathbf{w}_l^H\bm{\phi}_N(u)+\mathbf{w}_l^H\bar{\mathbf{n}}_l\nonumber\\
&=\sqrt{P_s}\alpha\mathbf{w}_l^H\bm{\phi}_N(u)+{n}_l.\label{eq:scalarmeas}
\end{IEEEeqnarray}The goal is to design $\mathbf{w}_l$ and infer $u$; $\mathbf{w}_l$ can be adapted over time to improve the inference. For a ULA with $\lambda/2$ inter-element spacing\footnote{$\lambda/2$ inter-element spacing prevents ambiguity in angular estimation.}, where $\lambda$ denotes the wavelength of the received signal, we have\begin{equation}
    \bm{\phi}_N(u)=[1\>\exp{(j\pi u)}\>\cdots\>\exp{(j\pi(N-1)u)}]^T.
\end{equation}In this paper, we consider the ULA geometry for easier exposition of ideas. However, these ideas can be extended to planar geometries as well, such as uniform rectangular arrays\cite{trees02}. Next, we describe a high-level structure that we impose when designing the beamformer, $\mathbf{w}_l$.
\subsection{Synthesis of Virtual Array Manifold (SVAM) Sensing}
The sensing methodology is inspired from SAR systems used in remote sensing and automotive radar (see \cite{moreira13,waldschmidt21} and references within them). In typical SAR systems, the sensor motion allows synthesis of larger aperture than physical antenna, which helps to improve resolution. In this work, we mimic the sensor motion by designing $\mathbf{w}_l$ appropriately. We exploit the coherence interval to synthesize virtual apertures \emph{over time}. An important consequence
is that, such measurements preserve phase information from the physical antenna, which captures rich information about the DoA of the incoming signal. The proposed sensing can be applied more broadly to multi-path angles. Moreover, leveraging the complex exponential structure present at the combiner output $y_l,l\in[L]$, it is possible to apply an \emph{unlimited} number of digital filters on these measurements.
\subsubsection{Constructing a Virtual ULA with $\lambda/2$ Inter-element Spacing}\label{sec:svamula}

Let $N_v$ denote the number of antenna elements in the \emph{virtual} ULA we wish to create. We assume that the total training duration, $L$, is divisible by $N_v$ for simplicity. Let $t(l)=\mathrm{floor}(l/N_v)$ denote\footnote{ $\mathrm{floor}(\cdot)$ denotes the floor function.} the (ULA) segment index (see Fig.~\ref{fig:numerology}). We design a beamformer of size $M=N-N_v+1$, $\mathbf{f}_{t(l)}\in\mathbb{C}^{M}$, such that $\lVert\mathbf{f}_{t(l)}\rVert_2=1$, initially spanning the region of interest (RoI). A RoI incorporates any prior information available about the DoA. In certain scenarios the receiver may be interested in identifying paths in a narrow region, for example, due to restricted mobility patterns for the transmitter. As the system gathers information about the unknown DoA, $u$, the design of the beamformer, $\mathbf{f}_{t(l)}$, is adapted. The analog combiner at instant $l$ is given by\begin{equation}
    \mathbf{w}_l=\left[\begin{array}{ccc}\mathbf{0}^T_{\mathrm{mod}(l,N_v)} & \mathbf{f}_{t(l)}^T& \mathbf{0}_{N_v-\mathrm{mod}(l,N_v)-1}^T\end{array}\right]^T.\label{eq:wlstructure}
\end{equation}Thus, within a segment duration, the beamformer slides along the antenna aperture.
In contrast to the work in \cite{pote19}, here only a single RF chain is available and thus a virtual ULA segment is synthesized over time. Let\begin{equation}\beta_{t(l)}(u)=\mathbf{f}_{t(l)}^H\bm{\phi}_{M}(u),\label{eq:betadef}\end{equation}denote the complex gain of the beamformer along the angle $u$. The signal $y_l$ post-combining can be expressed as\begin{equation}
    y_l=\mathbf{w}_l^H\mathbf{x}_l=\sqrt{P_s}\alpha\beta_{t(l)}(u)\cdot\exp{\left(j\pi u\mathrm{mod}(l,N_v)\right)}+n_l.\label{eq:scalarmeas_propsensing}
\end{equation}Note that the gain $\beta_{t(l)}(u)$ does not change within a segment, but `$\exp{\left(j\pi u\mathrm{mod}(l,N_v)\right)}$' varies within the segment.\begin{remark}
    The beamforming gain, measured in terms of $\vert\beta_{t(l)}(u)\vert^2$, in the passband depends primarily on the beamwidth of the beamformer, $\mathbf{f}_{t(l)}$. For an ideal beamformer design, the gain in the beamformer passband corresponding to a beamwidth of $\frac{2}{R},R\geq 1,$ in $u$-space is given by $\vert\beta_{t(l)}(u)\vert^2=R$. As the beamformer size, $M$, increases the beamformer response approaches the ideal response.
\end{remark}
We drop the notation for dependence of $t$ on $l$ for simplicity. We stack the measurements within a segment duration to form $\mathbf{y}_t=\left[
y_{tN_v} \>\> y_{tN_v+1} \>\> \cdots \>\> y_{(t+1)N_v-1}\right]^T\in\mathbb{C}^{N_v}, t\in[L/N_v]$,
\begin{IEEEeqnarray}{ll}
    \mathbf{y}_t
    &=\sqrt{P_s}\alpha\beta_{t}(u)\left[
         1\>
         \exp{(j\pi u)}
         \cdots
         \exp{(j\pi (N_v-1) u)}\right]^T+\mathbf{n}_t\nonumber\\
    &=\sqrt{P_s}\alpha\beta_{t}(u)\bm{\phi}_{N_v}(u)+\mathbf{n}_t.\label{eq:virtualmeas}
\end{IEEEeqnarray}We identify the following design parameters: a) $N_v\in\{1,2,\ldots,N\}$, the virtual ULA size, and b) beamformer, $\mathbf{f}_t$, design, which includes the beam direction and beamwidth. $N_v=1$ reduces to the conventional beam design.
Thus, the proposed sensing strategy includes the methodology adopted for sensing in \cite{alkhateeb14, chiu19} as a special case. The hierarchical codebook in \cite{alkhateeb14} designed using a least squared error criterion imposes a constant amplitude and phase in the passband. The inference is improved by relaxing the constant phase requirement in the passband. Thus, in this work we design the beamformers as \emph{linear}-phase Finite Impulse Response (FIR) filter using the Parks-McClellan algorithm\cite{oppenheim09}.

Let $\tilde{\mathbf{y}}_t=[\mathbf{y}_0^T,\ldots,\mathbf{y}_{t}^T]^T\in\mathbb{C}^{(t+1)N_v}$ denote the measurements until snapshot index, $l=(t+1)N_v-1$. Then\begin{IEEEeqnarray}{ll}
    \tilde{\mathbf{y}}_t=\left[\begin{array}{c}\mathbf{y}_0\\\mathbf{y}_1\\\vdots\\\mathbf{y}_t\end{array}\right]&=\sqrt{P_s}\alpha\left[\begin{array}{c}
       \beta_0(u)\bm{\phi}_{N_v}(u)\\
       \beta_1(u)\bm{\phi}_{N_v}(u)\\
       \vdots\\
       \beta_t(u)\bm{\phi}_{N_v}(u)
    \end{array}\right]+\left[\begin{array}{c}
        \mathbf{n}_0\\
        \mathbf{n}_1\\
        \vdots\\
        \mathbf{n}_t
    \end{array}\right]\nonumber\\
    &=\sqrt{P_s}\alpha\left(\tilde{\bm{\beta}}_t(u)\otimes\bm{\phi}_{N_v}(u)\right)+\tilde{\mathbf{n}}_t,\IEEEeqnarraynumspace\label{eq:virtualmeasstack}
\end{IEEEeqnarray}where $\tilde{\bm{\beta}}_t(u)=\left[\beta_0(u),\beta_1(u),\ldots,\beta_t(u)
\right]^T$ and $\mathbf{\tilde{n}}_t=[\mathbf{n}_0^T,\mathbf{n}_1^T,\ldots,\mathbf{n}_t^T]^T$. In Section~\ref{sec:alphaunknown} we discuss how to estimate $u$ along with $\alpha$ using $\tilde{\mathbf{y}}_t$, and how to adaptively design the beamformer i.e., $\mathbf{f}_{t+1}$ for the next segment.\begin{remark}
    It is important to highlight the significance of the proposed sensing methodology. Given just two measurements, it is possible to construct a virtual ULA under the proposed sensing with $N_v=2$. This is equivalent to a contrived \emph{single} snapshot measurement from a physical array of size $2$. Owing to the rich (array) geometrical information preserved in the measurements, it is thus possible to \emph{estimate} the dominant path DoA in a \emph{gridless} manner using existing techniques\cite{pote23doa}. In contrast, the beam scan operation using $N_v=1$ requires as many measurements as the codebook size to \emph{detect} the DoA.\end{remark} 
\subsubsection{Constructing a Virtual Sparse Linear Array}
The construction presented in the previous subsection can be extended to form virtual ULAs with more than $\lambda/2$ spacing\footnote{Any ambiguity in angular estimation can be resolved if the RoI is an appropriate fraction of the spatial region.}. More generally, a Sparse Linear Array (SLA) can also be realized as the virtual array geometry, for example, minimum redundancy arrays\cite{moffet68}, nested arrays\cite{pal10} or co-prime arrays\cite{vaidyanathan11}. These can help to increase the virtual aperture and improve resolution for the same segment duration. Let $N_v$ denote the number of antenna elements in the virtual SLA we wish to construct over time. Let $\mathbb{P}=\{P_i:0\leq P_i<N,P_i\in\mathbb{Z},i\in [N_v]\}$ denote the set of sensor positions in the SLA ordered in an increasing manner; $P_0=0$, without loss of generality. We design a beamformer, $\mathbf{f}_{t}$, of length $M=N-P_{N_v-1}$. The analog combiner at time $l$, in the case of SLA, is given by\begin{equation}
    \mathbf{w}_l=\left[\begin{array}{ccc}
    \mathbf{0}_{P_{\mathrm{mod}(l,N_v)}}^T & \mathbf{f}_{t}^T & \mathbf{0}_{N-M-P_{\mathrm{mod}(l,N_v)}}^T
    \end{array}\right]^T.
\end{equation}Using identical notation to describe the complex gain, $\beta_{t}(u)$, as in (\ref{eq:betadef}), the signal $y_l$ post-combining can be expressed as\begin{equation}
    y_l=\mathbf{w}_l^H\mathbf{x}_l=\sqrt{P_s}\alpha\beta_{t}(u)\cdot\exp{\left(j\pi uP_{\mathrm{mod}(l,N_v)}\right)}+n_l.\label{eq:scalarmeas_propsensing_sla}
\end{equation}Note that (\ref{eq:scalarmeas_propsensing_sla}) generalizes (\ref{eq:scalarmeas_propsensing}) for the SLA case. Finally, the measurements within the $t$-th SLA segment can be stacked as\begin{IEEEeqnarray}{ll}
\mathbf{y}_t&=\sqrt{P_s}\alpha\beta_t(u)[1\>\exp{(j\pi P_1u)}\cdots\exp{(j\pi P_{N_v-1}u)}]^T+\mathbf{n}_t\nonumber\\
&=\sqrt{P_s}\alpha\beta_t(u)\mathbf{S}_{\mathbb{P}}\bm{\phi}_N(u)+\mathbf{n}_t,
\end{IEEEeqnarray}where $\mathbf{S}_{\mathbb{P}}\in\mathbb{R}^{N_v\times N}$ is a binary sampling matrix
given by\begin{equation}
    [\mathbf{S}_{\mathbb{P}}]_{m,n}=\left\{\begin{array}{cc}
       1  &  \mbox{if }n=P_m\\
       0  & \mbox{otherwise}
    \end{array}\right.,m\in [N_v],n\in [N].
\end{equation}In the remainder of the work, we focus on the virtual ULA with $\lambda/2$ spacing-based sensing for ease of exposition, but the ideas presented can be easily extended to the virtual SLA case.

We refer to the  sensing methodology described in this section as Synthesis of Virtual Array Manifold (SVAM). Furthermore, we describe the SVAM sensing in conjunction with the virtual ULA size (with $\lambda/2$ spacing) by SVAM-$N_v$. For example, SVAM-$2$ indicates the sensing methodology employed is as described in subsection~\ref{sec:svamula} with $N_v=2$.
\subsection{Implications of Using SVAM sensing to multipath channels}The proposed SVAM sensing approach can be applied to scenarios including multiple paths. In this subsection, we briefly take a detour from (\ref{eq:aesignal}) and consider the following more general measurement model corresponding to a channel with $K$ paths ($K\geq 1$)\begin{equation}
    \mathbf{x}_l=\sqrt{P_s}\sum_{k=1}^K\alpha_k\bm{\phi}_N(u_k)+\bar{\mathbf{n}}_l,\quad l\in [L],\label{eq:aesignalmultipath}
\end{equation}where the notations simply extend for $K$ paths compared to (\ref{eq:aesignal}) and $P_s$ denotes the average power value for the $K$ paths. A similar development as shown in Subsection~\ref{sec:svamula} leads to the following received measurement vector after beamforming in the $t$-th virtual ULA segment duration\begin{equation}
\mathbf{y}_t=\sqrt{P_s}\sum_{k=1}^K\alpha_k\beta_t(u_k)\bm{\phi}_{N_v}(u_k)+\mathbf{n}_t\end{equation}compared to (\ref{eq:virtualmeas}). The combined measurements after $(t+1)N_v$ snapshots is given by\begin{IEEEeqnarray}{ll}
    \tilde{\mathbf{y}}_t=\left[\begin{array}{c}\mathbf{y}_0\\\mathbf{y}_1\\\vdots\\\mathbf{y}_t\end{array}\right]&=\sqrt{P_s}\sum_{k=1}^K\alpha_k\left(\tilde{\bm{\beta}}_t(u_k)\otimes\bm{\phi}_{N_v}(u_k)\right)+\tilde{\mathbf{n}}_t,\IEEEeqnarraynumspace\label{eq:virtualmeasstackKpaths}
\end{IEEEeqnarray}compared to (\ref{eq:virtualmeasstack}), where $\tilde{\bm{\beta}}_t(u)=\left[\beta_0(u),\beta_1(u),\ldots,\beta_t(u_k)
\right]^T$ and $\mathbf{\tilde{n}}_t=[\mathbf{n}_0^T,\mathbf{n}_1^T,\ldots,\mathbf{n}_t^T]^T$. The beamforming gain for each path is a function of the beamformer design $\mathbf{f}_t$ and the path angle. Thus, all the paths within the passband of the beamformer $\mathbf{f}_t$
are boosted. This demonstrates the applicability of the proposed SVAM sensing for more general channel models beyond the scenario considered in (\ref{eq:aesignal}). The remainder of this work specializes to the single dominant path model in (\ref{eq:aesignal}) for tractability and ease of exposition. The more general scenarios involving multipaths (such as \cite{va16,song19,yang20,khordad23}) is left as future work. 
\section{Benefits of SVAM For Beam Alignment
}\label{sec:alphaknown}
We discuss some of the benefits of the proposed SVAM sensing for estimating the DoA, $u$. We demonstrate the benefits in two settings: i. agnostic to the adaptive scheme used, ii. when hierarchical posterior matching (hiePM)\cite{chiu19} scheme is used. In both settings we assume that $\alpha$ is known. We also briefly discuss the role of virtual antenna size $N_v$
for improving the estimation performance. The more practical scenario, where $\alpha$ is unknown, is discussed in the next section.

\subsection{Adaptive Scheme-Agnostic Analysis}\label{sec:alphaknownCRBbenchscheme}
 We begin by first deriving the CRB after $L$ snapshots, assuming $\alpha$ is known. Let $\mathbf{W}=[\mathbf{w}_0,\ldots,\mathbf{w}_{L-1}]\in\mathbb{C}^{N\times L},\lVert\mathbf{w}_l\rVert_2=1$ denote the matrix of beamformers used to generate measurements $\mathbf{y}=[y_0,\ldots,y_{L-1}]^T\in\mathbb{C}^L$ as in (\ref{eq:scalarmeas}). Note that $\mathbf{w}_l$'s may be designed generally, and not necessarily under SVAM for the following result to hold.
 \begin{theorem}
     The $\mathrm{CRB}(u)$ on the variance for estimating $u$ using the beamformer matrix $\mathbf{W}=[\mathbf{w}_0,\ldots,\mathbf{w}_{L-1}],\lVert\mathbf{w}_l\rVert_2=1$ as in (\ref{eq:scalarmeas}) over $L$ snapshots when $\alpha$ is known is given by\begin{equation}
         \mathrm{CRB}(u)=\frac{\sigma_n^2}{2P_s\vert\alpha\vert^2}\left\{\left(\frac{\partial}{\partial u}\bm{\phi}_N(u)\right)^H\mathbf{W}\mathbf{W}^H\frac{\partial}{\partial u}\bm{\phi}_N(u)\right\}^{-1}.\label{eq:crbalphaknowngeneral}\end{equation}\label{thm:CRBalphaknown}
 \end{theorem}\begin{proof}
     The proof follows standard steps for deriving CRB and provided in \cite{pote23dissert}. 
 \end{proof}In this subsection, we compare two sensing strategies - both deterministically modify the beamformer design after every $N_v$ snapshots. The proposed sensing strategy involves a shift in space as described in (\ref{eq:wlstructure}), and requires the beamformer to be fixed for $N_v$ snapshots by design. The alternative (\emph{benchmark}) strategy designs the beamformer $\mathbf{w}_l$ of size $N$ without inserting $0$'s - used in SVAM to effect a \emph{linear time invariant} operation. The beamformer for the two strategies considered here are to be designed with identical specifications except the length. The SVAM beamformer $\mathbf{f}_t$ has a size of $M (=N-N_v+1)$, whereas the alternative strategy utilizes the total antenna aperture of size $N$. We expect that for a large antenna size $N$, which is typical in mmWave systems, the slightly different length of the beamformers will have negligible impact. The rest of the beamformer specifications, which can change every $N_v$ snapshots, may be chosen arbitrarily for the discussion in this subsection. We defer the discussion that involves using the specific adaptive scheme - HiePM\cite{chiu19} to the subsection~\ref{sec:numsecalphaknown}. We specialize the CRB expression in Theorem~\ref{thm:CRBalphaknown} for the benchmark strategy and the proposed sensing strategy. This exercise helps to understand: \emph{how informative are the measurements available post analog combining about the unknown DoA using either of the two techniques for designing the analog combiners?} For both the cases, we treat as if the same ordered set of $L$ received signal snapshots $\mathbf{x}_l\in\mathbb{C}^N,l\in[L]$ were available at the antenna.

\subsubsection{CRB for the Benchmark Strategy}Let $\mathbf{W}^\mathrm{B}=[\mathbf{w}^{\mathrm{B}}_0,\mathbf{w}^{\mathrm{B}}_1,\ldots,\mathbf{w}^{\mathrm{B}}_{L-1}]\in\mathbb{C}^{N\times L}$ be the $L$ beamformers  used to generate the measurements as in (\ref{eq:scalarmeas}); superscript $\mathrm{B}$ highlights the benchmark sensing strategy. Note that\begin{equation}
\mathbf{w}^{\mathrm{B}}_l=\mathbf{w}^{\mathrm{B}}_{N_v\times\mathrm{floor}(l/N_v)},\label{eq:benchmkconstraint}
\end{equation}under the benchmark scheme. Let $\mathbf{F}^{\mathrm{B}}=[\mathbf{w}^{\mathbf{B}}_0,\mathbf{w}^{\mathbf{B}}_{N_v},\ldots,\mathbf{w}^{\mathbf{B}}_{L-N_v}]\in\mathbb{C}^{N\times\frac{L}{N_v}}$ denote the matrix of unique beamformers from $\mathbf{W}^{\mathrm{B}}$. We have the following result.\begin{corollary}\label{corr:benchsens}
     The $\mathrm{CRB}(u)$ using measurements in (\ref{eq:scalarmeas}) from the beamformer matrix $\mathbf{W}^{\mathrm{B}}$
     under the constraint in (\ref{eq:benchmkconstraint}) over $L$ snapshots when $\alpha$ is known is given by\begin{IEEEeqnarray}{ll}
         &\mathrm{CRB}(u)\nonumber\\
         &=\frac{1}{N_v}\frac{\sigma_n^2}{2P_s\vert\alpha\vert^2}\left\{\left(\frac{\partial}{\partial u}\bm{\phi}_N(u)\right)^H\mathbf{F}^{\mathrm{B}}(\mathbf{F}^{\mathrm{B}})^H\frac{\partial}{\partial u}\bm{\phi}_N(u)\right\}^{-1}.\label{eq:crbbenchsens}
     \end{IEEEeqnarray}
\end{corollary}\begin{proof}
Proof follows from simplifying (\ref{eq:crbalphaknowngeneral}) using (\ref{eq:benchmkconstraint}), and is provided in \cite{pote23dissert}.
\end{proof}

\subsubsection{CRB for the Proposed Sensing Strategy}Let $\mathbf{W}^\mathrm{P}=[\mathbf{w}^{\mathrm{P}}_0,\mathbf{w}^{\mathrm{P}}_1,\ldots,\mathbf{w}^{\mathrm{P}}_{L-1}]\in\mathbb{C}^{N\times L}$ be the $L$ beamformers designed under SVAM as described in (\ref{eq:wlstructure}); superscript $\mathrm{P}$ highlights the proposed SVAM sensing. Let $\mathbf{F}^{\mathrm{P}}=[\mathbf{f}_0,\mathbf{f}_1,\ldots,\mathbf{f}_{L/N_v-1}]\in\mathbb{C}^{M\times\frac{L}{N_v}}$ denote the matrix of SVAM beamformers
. 
We have the following result.\begin{corollary}\label{corr:svamsens}
     The $\mathrm{CRB}(u)$  using measurements in (\ref{eq:scalarmeas}) from the beamformer matrix $\mathbf{W}^{\mathrm{P}}$
     under the construction in (\ref{eq:wlstructure}) over $L$ snapshots when $\alpha$ is known is given by\begin{IEEEeqnarray}{ll}
         \mathrm{CRB}(u)&=\frac{1}{N_v}\frac{\sigma_n^2}{2P_s\vert\alpha\vert^2}\left\{\left(\frac{\partial}{\partial u}\bm{\phi}_M(u)\right)^H\mathbf{F}^{\mathrm{P}}(\mathbf{F}^{\mathrm{P}})^H\frac{\partial}{\partial u}\bm{\phi}_M(u)\right.\nonumber\\
         &\quad+\Biggl.G\Biggr\}^{-1},\mbox{ where}\label{eq:crbsvamsens}
     \end{IEEEeqnarray}\begin{IEEEeqnarray}{ll}G&=\frac{\pi^2(N_v-1)(2N_v-1)}{6}\bm{\phi}_{M}(u)^H\mathbf{F}^{\mathrm{P}}(\mathbf{F}^{\mathrm{P}})^H\bm{\phi}_{M}(u)\nonumber\\
     &\quad-\pi(N_v-1)\mathrm{Im}\left\{\left(\frac{\partial}{\partial u}\bm{\phi}_M(u)\right)^H\mathbf{F}^{\mathrm{P}}(\mathbf{F}^{\mathrm{P}})^H\bm{\phi}_M(u)\right\}\IEEEeqnarraynumspace\label{eq:Gexp}\end{IEEEeqnarray}
\end{corollary}\begin{proof}
Proof follows from simplifying (\ref{eq:crbalphaknowngeneral}) using (\ref{eq:wlstructure}), and is provided in \cite{pote23dissert}.
\end{proof}The CRB expression in (\ref{eq:crbsvamsens}) differs from the expression in (\ref{eq:crbbenchsens}) in two ways. The first term in (\ref{eq:crbsvamsens}) involves an array of dimension $M=N-N_v+1$ instead of $N$ in (\ref{eq:crbbenchsens}). For large array sizes, $N$, which are typical in mmWave systems, we expect this term to be similar to that in (\ref{eq:crbbenchsens}).  Secondly, the denominator in (\ref{eq:crbsvamsens}) has an additional second term `$+G$'.
We show that the conditions under which $G\geq 0$ is not difficult to satisfy, by deriving a sufficient condition to ensure the same.\begin{theorem}\label{thm:Gnonneg}
    $G$ in (\ref{eq:Gexp}) is non-negative if\begin{equation}
        \frac{\bm{\phi}^H_M(u)\mathbf{F}^{\mathrm{P}}(\mathbf{F}^{\mathrm{P}})^H\bm{\phi}_M(u)}{\lVert\bm{\phi}_M(u)\rVert^2}\geq\frac{\lambda_{\mathrm{max}}\left((\mathbf{F}^{\mathrm{P}})^H\mathbf{P}_{u,\perp}\mathbf{F}^{\mathrm{P}}\right)}{4},\label{eq:Gnonnegcond}
    \end{equation}where $\lambda_{\mathrm{max}}(\mathbf{X})$ denotes the largest eigenvalue of the matrix $\mathbf{X}$. $\mathbf{P}_{u,\perp}=[\bm{\phi}_{M}(u)\> \bm{\phi}^{\perp}_{M}(u)]\left[\begin{array}{cc}
       \lVert\bm{\phi}_{M}(u)\rVert^2  &  0\\
       0  & \lVert\bm{\phi}^{\perp}_{M}(u)\rVert^2 
    \end{array}\right]^{-1}$ $\times[\bm{\phi}_{M}(u)\>\> \bm{\phi}^{\perp}_{M}(u)]^H$ denotes a projection onto the subspace of orthogonal vectors $\bm{\phi}_M(u)$ and $\bm{\phi}^{\perp}_M(u)=\left[\begin{array}{cccc}-\frac{(M-1)}{2}&(1-\frac{(M-1)}{2})&\ldots&\frac{(M-1)}{2}\end{array}\right]^T\odot\bm{\phi}_M(u)$. 
\end{theorem}\begin{proof}The proof is provided in Appendix section~\ref{sec:Gnonnegproof}.
\end{proof}\noindent The following remark discusses the implication of Theorem~\ref{thm:Gnonneg}.\begin{remark}
    If the left singular vectors of $\mathbf{F}^{\mathrm{P}}$ includes $\bm{\phi}_M(u)$ and $\bm{\phi}^{\perp}_M(u)$ (post-normalization), we can further simplify (\ref{eq:Gnonnegcond}). If $\bm{\phi}_M(u)$ leads $\bm{\phi}^{\perp}_M(u)$, then (\ref{eq:Gnonnegcond}) is trivially satisfied. If the opposite is true, in that, $\bm{\phi}^{\perp}_M(u)$ leads $\bm{\phi}_M(u)$, then (\ref{eq:Gnonnegcond}) describes the required gap in the two corresponding singular values within $\mathbf{F}^{\mathrm{P}}$. In practice, we expect an adaptive scheme to choose beamformers close to the direction of the DoA. Also, it was found out that $G\geq 0$ very often even when $\mathbf{F}^{\mathrm{P}}$ contained i.i.d. complex Gaussian random entries.
\end{remark}
\begin{figure*}
    \centering
    \begin{tabular}{ccc}
        \includegraphics[width=0.3\linewidth]{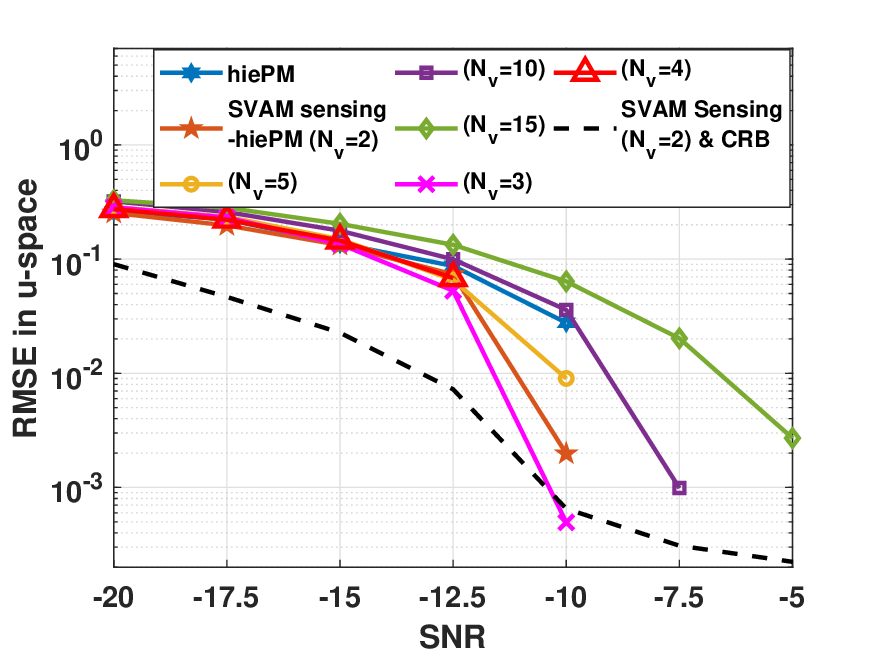}
        & \includegraphics[width=0.3\linewidth]{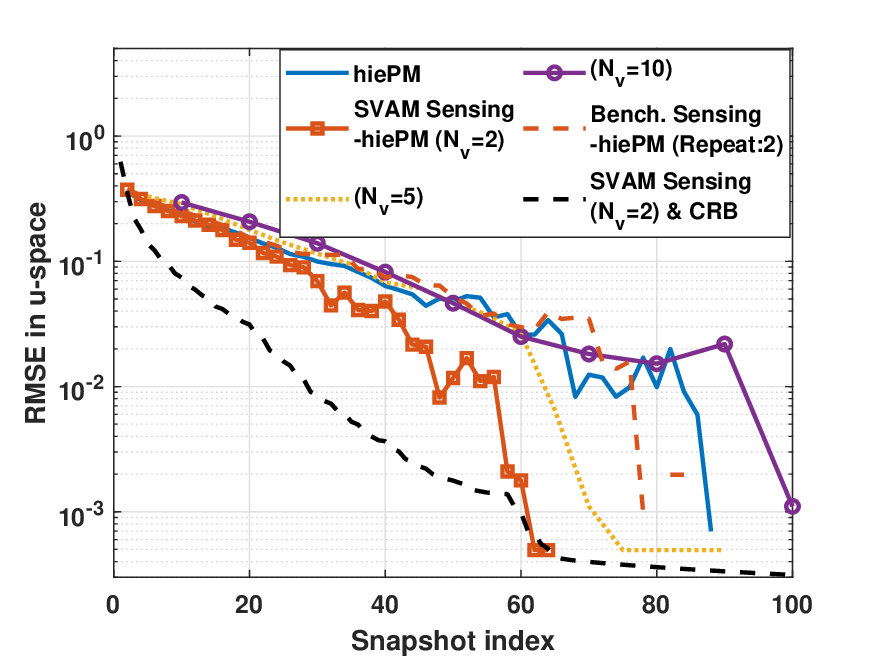} & \includegraphics[width=0.3\linewidth]{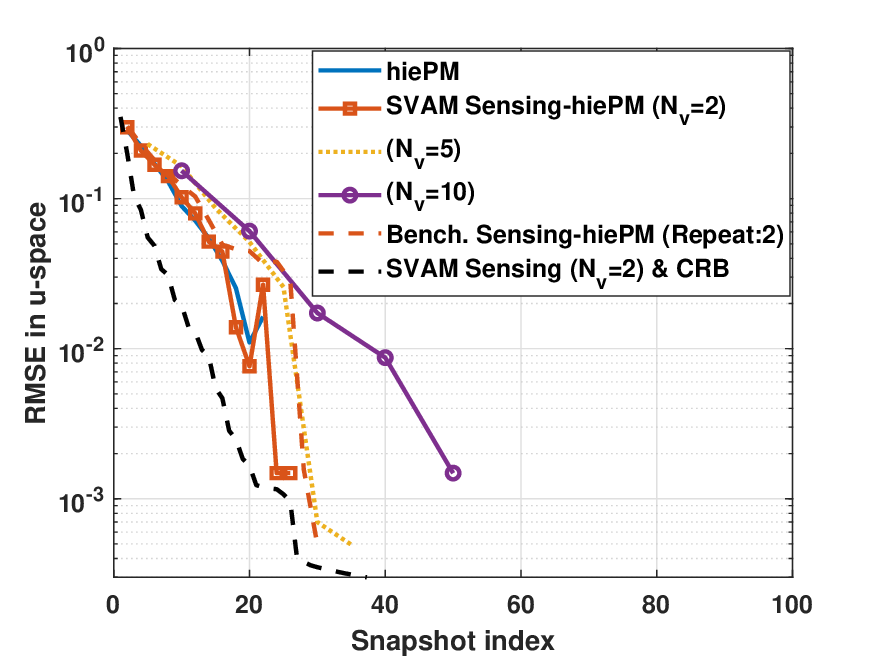}\\
        (a) Snapshots, $L=60$ & (b) SNR = $-10$ dB & (c) SNR = $-5$ dB
    \end{tabular}
    \caption{
    (a) RMSE vs SNR: Proposed SVAM sensing with modified-hiePM improves over hiePM\cite{chiu19} when $N_v\in\{2,3,4,5\}$. (b) RMSE vs number of snapshots at SNR=$-10$ dB: Proposed sensing reduces the duration needed for beam alignment when $N_v=2$. (c) RMSE vs number of snapshots at SNR=$-5$ dB: At high SNR, the proposed scheme still improves over the benchmark scheme, but takes slightly more time to converge.}\vspace{-1em}
    \label{fig:alphaknown}
\end{figure*}

 Theorem~\ref{thm:Gnonneg} emphasizes that the sufficient condition for ensuring $G\geq0$ are not difficult to satisfy and highlights the benefit from using SVAM sensing
 compared to the benchmark scheme. The benchmark analysis can be extended to include more general scenario, wherein it need not repeat the beamformer in $N_v$ snapshots. Since the contribution from each snapshot appears independent of other snapshots, as a linear sum in the denominator in the general CRB expression in (\ref{eq:crbalphaknowngeneral}), any time a beamformer is repeated in a benchmark scheme, it can be replaced with the proposed scheme. A similar impact and analysis can be carried out to reveal ensuing benefits.
 \subsection{Case Study: Combining SVAM with HiePM Framework}\label{sec:alphaknownsvamhiePM}
The hiePM algorithm\cite{chiu19}  processes each new snapshot and updates the beamformer based on the current estimate of the posterior density on the unknown DoA. To incorporate the proposed sensing, one approach is to update the beamformer after every $N_v$ snapshots, where $N_v$ denotes the size of the virtual ULA. Within each $N_v$ interval, a SVAM beamformer, $\mathbf{f}_t$, is designed; $\mathbf{f}_t$ is simply a codeword from the hierarchical codebook that satisfies the selection criteria within hiePM framework. The beamfomer (of physical antenna size, $N$) for $N_v$ snapshots within this interval is constructed as in (\ref{eq:wlstructure}). The remainder steps in Algorithm 1 in \cite{chiu19} are compatible with the proposed sensing.
This modification also becomes crucial when $\alpha$ is unknown; this is discussed in Section~\ref{sec:alphaunknown}.
\begin{remark}
    The impact of the modification can be understood in the following manner. The hiePM strategy in \cite{chiu19} may repeat the same beamformer multiple times until the posterior condition triggers a new beamformer. While it repeats the beamformer, it gains only in terms of the Signal-to-Noise Ratio (SNR), which is also evident from the CRB analysis for the benchmark scheme in (\ref{eq:crbbenchsens}). $N_v$ multiplied with $P_s$ in (\ref{eq:crbbenchsens}) indicates a $N_v$-fold SNR boost. Instead of a repeated beamformer, the proposed sensing aims to achieve both, a SNR boost and at the same time an \emph{(virtual) aperture} gain (`$+G$' term in the CRB analysis). Since the adaptive beamformer sequence is unknown apriori, analysis of the proposed modification to hiePM in this subsection is more involved. Instead, we provide empirical studies in the next subsection.  
\end{remark}\begin{remark}
     The proposed sensing can be incorporated in other beam alignment procedures as well. The inference procedure and adaptive strategy may require suitable modifications, or be left unaltered as demonstrated using the hiePM case study in this subsection. 

\end{remark}
\subsection{Numerical Results}\label{sec:numsecalphaknown}
We numerically analyze the benefits of the proposed sensing when used with the hiePM framework in \cite{chiu19}, and when $\alpha$ is assumed to be known.
We compare the hiePM\cite{chiu19} approach with i. the proposed modified algorithm in Section~\ref{sec:alphaknownsvamhiePM}, and ii. the benchmark scheme where the beamformer is simply repeated, as described in Section~\ref{sec:alphaknownCRBbenchscheme}. We also plot the CRB\footnote{Note that since the posterior density is computed on a grid which includes the ground truth angle, the performance may be biased. Thus, CRB is not a valid lower bound but it is provided for insight.} conditioned on the beamformers selected by the modified algorithm with $N_v=2$. We assume that the DoA, $u$, lies in between $[0,1)$ i.e., RoI is $\frac{1}{2}$ of the entire space. The information about the RoI is incorporated into the posterior calculation. We set the following parameters as: Physical antenna size: $N=64$, number of grid points uniformly spaced in RoI: $G=64$, number of random realizations for averaging: $Q=1000$. The main metric employed for comparison is the root mean squared error (RMSE) computed in $u$-space 
as $
    \mbox{RMSE}=\sqrt{\frac{1}{Q}\sum_{q=1}^Q\left(\hat{u}_{q}-u_{q}\right)^2}
$, $u_q$ and $\hat{u}_{q}$ denotes the true and the estimated angle respectively, in the $q-$th realization. The estimated angle is the on-grid angle corresponding to the mode of the posterior density estimated by the algorithm for each curve. In this work we define SNR as the signal power to noise ratio at each antenna i.e., before combining.
\subsubsection{Performance as a function of SNR}
In Fig.~\ref{fig:alphaknown} (a) we plot the RMSE as a function of SNR for the modified-hiePM under the proposed SVAM sensing using different virtual ULA sizes. As observed in this plot, given $L=60$ snapshots, the curves using $N_v\in\{2,3,4,5\}$ improve over the hiePM in \cite{chiu19}. As the virtual aperture increases, the number of beamformer updates, given by `$L/N_v$', reduces. For $N_v\in\{10,15\}$, the reduced number of updates or longer acquisition time before a beamformer update, is seen to negatively impact the performance. Studying the tradeoff between the virtual antenna size and the frequency of update as they impact the beam alignment performance is an important future direction. 
\subsubsection{Performance as a function of number of snapshots}
In Fig.~\ref{fig:alphaknown} (b) and (c), we plot the RMSE as a function of number of snapshots at SNR=$-10$ dB and $-5$ dB, respectively. As observed from the Fig.~\ref{fig:alphaknown} (b), the performance of hiePM improves under the proposed sensing when $N_v$ is set to $2$. Setting it to a higher value such as $N_v=10$ degrades the performance, for reasons similar to those discussed for Fig.~\ref{fig:alphaknown} (a). In red dashed curve we plot the performance for hiePM with benchmark sensing scheme (in Section~\ref{sec:alphaknownCRBbenchscheme}). Although the benchmark scheme improves the beam alignment duration over hiePM (solid blue curve), it still is much higher compared to setting $N_v=2$ in the proposed sensing. This indicates that the improvement under SVAM is not merely due to more reliable beamformer adaptations caused by basing adaptations on more measurements. It is inherent to the \emph{phase response} characteristics of the beamformer under SVAM, as described in (\ref{eq:scalarmeas_propsensing}).
Note that as we increase SNR to $-5$ dB, the performance gains can be limited as seen from Fig.~\ref{fig:alphaknown}(b). This highlights that SVAM is beneficial in reducing the training time at low SNR. Also it is again observed that SVAM improves over the benchmark scheme. Thus, a 
dynamic method to adapt the virtual ULA size is an interesting future direction.   
\section{Beam Alignment Algorithm with Unknown $\alpha$}
\label{sec:alphaunknown}
We begin with a wide RoI, $u\in[u_l,u_r],u_l<u_r,u_l,u_r\in[-1,1)$, and presume that the DoA lies within the RoI. Therefore, we initially set the SVAM beamformer to span this region and suppress any interference coming from outside the RoI. This ensures that we incorporate any prior information available about the DoA. The approach is also practical, as base stations are typically deployed with dedicated antennas to serve a specific RoI. We collect measurements over time and process them to compute an approximate posterior on the unknown angle. We adapt SVAM beamformer once we accrue enough posterior mass around the mode of the posterior. We now describe the methodology adopted in this work to estimate the posterior and to adapt the SVAM beamformer over time.
\subsection{Algorithm Preliminaries: Initializing Angular Grid and Stochastic Modeling of both $\alpha$ and $u$}
Since $\alpha$ is unknown but assumed to be fixed during the training phase, it can be estimated along with the DoA given \emph{two} or more measurements. We observe this from the conditional CRB on variance for angle estimation when $\alpha$ is unknown. Let $\mathbf{W}=[\mathbf{w}_0,\ldots,\mathbf{w}_{L-1}]\in\mathbb{C}^{N\times L}$ denote the beamformer matrix used to generate measurements $\mathbf{y}=[y_0,\ldots,y_{L-1}]^T\in\mathbb{C}^L$ as in (\ref{eq:scalarmeas}). The beamformers may be designed generally, and not necessarily under SVAM. The $\mathrm{CRB}(u)$ is given by\begin{IEEEeqnarray}{ll}
    &\mathrm{CRB}(u)\nonumber\\
    &=\frac{\sigma_n^2}{2P_s\vert\alpha\vert^2}\left\{\left(\frac{\partial}{\partial u}\bm{\phi}_N(u)\right)^H\mathbf{W}\right.\nonumber\\
    &\>\>\left.\times\left(\mathbf{I}-\frac{\mathbf{W}^H\bm{\phi}_N(u)\bm{\phi}_N(u)^H\mathbf{W}}{\bm{\phi}_N(u)^H\mathbf{W}\mathbf{W}^H\bm{\phi}_N(u)}\right)\mathbf{W}^H\frac{\partial}{\partial u}\bm{\phi}_N(u)\right\}^{-1},\IEEEeqnarraynumspace
\end{IEEEeqnarray}and can be derived from the conditional CRB expression (Theorem 4.1) in \cite{stoica89} by stacking all $L$ snapshots along a column and replacing $\bm{\phi}_N(u)$ and $\frac{\partial}{\partial u}\bm{\phi}_N(u)$ by $\mathbf{W}^H\bm{\phi}_N(u)$ and $\mathbf{W}^H\frac{\partial}{\partial u}\bm{\phi}_N(u)$, respectively. For $L=1$, we can see that the CRB is infinite. Furthermore, if the beamformer is kept fixed over time i.e., if $\mathbf{w}_l=\mathbf{w}_0,l\in[L]$, or in general if $\mathbf{W}$ is rank-one, even then the CRB is infinite. This exercise highlights a key requirement for being able to estimate both DoA and $\alpha$ simultaneously. The magnitude or phase response of the beamformer $\mathbf{w}_l$ should have variation over time, and it must vary differently for different angles.
The requirement is not satisfied in the case of hiePM algorithm, if initial measurements are all taken using a fixed codeword, which is possible, early on, within the hiePM framework\footnote{The CRB analysis is applicable when the unknowns $(\alpha,u)$ are estimated in the maximum likelihood sense.
The issues can be circumvented by imposing a grid or strong prior information on $\alpha$\cite{ronquillo19}.}.
In the proposed sensing scheme this condition is naturally prevented, as the phase response (over time) leads to the virtual array manifold, which is sufficient to estimate both $\alpha$ and $u$ when $N_v>1$ and given at least two measurements.

The two basic pre-requisities for adapting the SVAM beamformer are the next beam direction to steer towards, and the beamwidth. Since the problem-at-hand considers the single path scenario, the beam direction can be estimated by \emph{matching}-based criteria which results in a maximum likelihood estimate (MLE), and it also yields a \emph{deterministic} estimate for $\alpha$. However, for a good beamwidth selection we need to account for the uncertainty in the estimation procedure. The bounds for the CRB analysis can be used. In this work, we achieve this by modeling both $\alpha$ and the unknown DoA, $u$, as \emph{stochastic} variables. For the latter, we impose a uniform prior distribution in the wide RoI. We model $\alpha$ with a complex circular Gaussian distribution with a parameterized prior. 

We introduce a uniform grid, $\mathbf{u}_{\mathrm{grid}}$, of size $G$
within the RoI, $[u_l,u_r]$. Let $u_i,i\in[G]$, denote the $i$-th grid point. The initial grid can be refined as we successively reduce uncertainty of the estimate for $u$ over snapshots. The grid refinement aspect is not focused in this paper, and left as future work. For each candidate DoA, $u_i$, we estimate a corresponding value for the complex path gain; let us denote the same as $\alpha_i$. We impose the following parameterized prior on $\alpha_i$\begin{equation}
    \alpha_i\sim\mathcal{CN}(0,\gamma_i).\label{eq:alphaprior}
\end{equation}
\begin{remark}
    Assigning a different complex gain, $\alpha_i$, per candidate angle, $u_i$, is a non-trivial choice. By doing so, we allow each candidate $u_i$ to \emph{explain} the measurements at its best. Going forward, we explicitly utilize the presence of a single path by identifying suitable value for hyperparameters, $\gamma_i$, separately, instead of jointly. A joint optimization is widely adopted in the \emph{absence} of knowledge of the model order\cite{wipf04}.
\end{remark}A consequence of the prior imposed in (\ref{eq:alphaprior}) is that the marginalized pdf of the measurements $\tilde{\mathbf{y}}_t$ in (\ref{eq:virtualmeasstack}) after $(t+1)N_v$ snapshots, conditioned on angle $u=u_i$ is given by
\begin{IEEEeqnarray}{ll}
    &f(\tilde{\mathbf{y}}_t\mid u=u_i; \gamma_i)\nonumber\\
    &\quad=\int f(\tilde{\mathbf{y}}_t\mid u=u_i,\alpha_i)f(\alpha_i;\gamma_i)\>d\alpha_i\nonumber\\
    &\quad=\frac{1}{\pi^{(t+1)N_v}\mathrm{det}(\bm{\Sigma}_{i,t})}\exp{\left(-\tilde{\mathbf{y}}_t^H\bm{\Sigma}_{i,t}^{-1}\tilde{\mathbf{y}}_t\right)}
    ,\mbox{ where}\IEEEeqnarraynumspace
\end{IEEEeqnarray}\begin{IEEEeqnarray}{ll}
    &\bm{\Sigma}_{i,t}\nonumber\\
    &=P_s\gamma_i\left(\tilde{\bm{\beta}}_t(u_i)\otimes\bm{\phi}_{N_v}(u_i)\right)\left(\tilde{\bm{\beta}}_t(u_i)\otimes\bm{\phi}_{N_v}(u_i)\right)^H+\sigma_n^2\mathbf{I}\nonumber\\
    &=P_s\gamma_i\left(\tilde{\bm{\beta}}_t(u_i)\tilde{\bm{\beta}}_t(u_i)^H\right)\otimes\left(\bm{\phi}_{N_v}(u_i)\bm{\phi}_{N_v}(u_i)^H\right)+\sigma_n^2\mathbf{I},\IEEEeqnarraynumspace\label{eq:Sigmait}
\end{IEEEeqnarray}where in the last line, we use the mixed product property of Kronecker product, namely if matrix products $\mathbf{AC}$ and $\mathbf{BD}$ are defined, then $(\mathbf{A}\otimes\mathbf{B})(\mathbf{C}\otimes\mathbf{D})=\mathbf{AC}\otimes\mathbf{BD}$. Note that $\sigma_n^2$ is assumed to be known\footnote{This assumption is also made in other works\cite{chiu19,ronquillo19}.}. In the absence of such knowledge, it can be included in the estimation procedure (see \cite{wipf04}). We focus on the estimation of $\alpha$ in this work, and study the impact of imperfect knowledge of $\sigma_n^2$ in the simulation section.
\subsection{Estimation of Hyperparameters under the MLE Framework}
We find the hyperparameters of the imposed prior, namely $\gamma_i,i\in[G]$, in the MLE sense. 
The derivation is simple, and similar to the work in \cite{pote23,tipping03}.
Since $\bm{\Sigma}_{i,t}$ in (\ref{eq:Sigmait}) can be expressed as a rank-one perturbation to the noise covariance matrix, the determinant can be expressed in closed-form as\begin{equation}
    \mathrm{det}(\bm{\Sigma}_{i,t})=(P_s\gamma_ig_t(u_i)\lVert\bm{\phi}_{N_v}(u_i)\rVert_2^2+\sigma_n^2)\cdot(\sigma_n^2)^{(t+1)N_v-1},\label{eq:detsigma}
\end{equation}where $g_t(u_i)=\sum_{t'=0}^t\vert\beta_{t'}(u_i)\vert^2$. Also, using matrix-inversion lemma, we can simplify $\tilde{\mathbf{y}}_t^H\bm{\Sigma}_{i,t}^{-1}\tilde{\mathbf{y}}_t$ as\begin{IEEEeqnarray}{ll}
    \tilde{\mathbf{y}}_t^H\bm{\Sigma}_{i,t}^{-1}\tilde{\mathbf{y}}_t&=\sigma_n^{-2}\lVert\tilde{\mathbf{y}}_t\rVert_2^2-\frac{P_s\gamma_i}{\sigma_n^{2}}\nonumber\\&\quad\times\frac{\vert\bm{\phi}^H_{N_v}(u_i)\sum_{t'=0}^t\beta^c_{t'}(u_i)\mathbf{y}_{t'}\vert^2}{(P_s\gamma_ig_t(u_i)\lVert\bm{\phi}_{N_v}(u_i)\rVert_2^2+\sigma_n^2)}.\label{eq:ysigmay}
\end{IEEEeqnarray}Using (\ref{eq:detsigma}) and (\ref{eq:ysigmay}), we maximize $f(\tilde{\mathbf{y}}_t\mid u=u_i;\gamma_i)$ over $\gamma_i$ by setting the derivative w.r.t. the same as zero. We get\begin{IEEEeqnarray}{ll}
    \hat{\gamma}_{i,t}=\max\Biggl\{&0,\frac{1}{P_sg_t(u_i)\lVert\bm{\phi}_{N_v}(u_i)\rVert_2^2}\Bigr.\nonumber\\
    &\left.\times\left(\left\vert\frac{\left(\tilde{\bm{\beta}}_t(u_i)\otimes\bm{\phi}_{N_v}(u_i)\right)^H}{\lVert\tilde{\bm{\beta}}_t(u_i)\otimes\bm{\phi}_{N_v}(u_i)\rVert_2}\tilde{\mathbf{y}}_{t}\right\vert^2-\sigma_n^2\right)\right\}.\IEEEeqnarraynumspace\label{eq:gammaupdate}
\end{IEEEeqnarray}The above can be interpreted as normalized beamforming output power that is compensated for the noise power. The beamformer in this case is the conventional beamformer\cite{trees02} and takes into account the complex gain, $\tilde{\bm{\beta}}_t(u_i)$, from the SVAM beamformer $\mathbf{f}_{t'},t'\in[t+1]$. The posterior density of $\alpha_i$ is also complex circular Gaussian, and can be computed as\begin{equation}
    f(\alpha_i\mid\tilde{\mathbf{y}}_t,u_i;\hat{\gamma}_{i,t})=\frac{1}{\pi\hat{\sigma}^2_{\alpha_i,t}}\exp{\left(-\frac{\vert\alpha_i-\hat{\mu}_{\alpha_i,t}\vert^2}{\hat{\sigma}^2_{\alpha_i,t}}\right)},
\end{equation}\begin{IEEEeqnarray}{ll}
    \begin{aligned}\mbox{where }\hat{\mu}_{\alpha_i,t}&=\sqrt{P_s}\hat{\gamma}_{i,t}\frac{\left(\tilde{\bm{\beta}}_t(u_i)\otimes\bm{\phi}_{N_v}(u_i)\right)^H\tilde{\mathbf{y}}_t}{P_s\hat{\gamma}_{i,t}g_t(u_i)\lVert\bm{\phi}_{N_v}(u_i)\rVert_2^2+\sigma_n^2}\\
    \hat{\sigma}^2_{\alpha_i,t}&=\hat{\gamma}_{i,t}\frac{\sigma_n^2}{P_s\hat{\gamma}_{i,t}g_t(u_i)\lVert\bm{\phi}_{N_v}(u_i)\rVert_2^2+\sigma_n^2}.\end{aligned}\IEEEeqnarraynumspace\label{eq:alphaposterior}
\end{IEEEeqnarray}       
\subsection{Computing Posterior on $u$}
We begin with describing the posterior probability, $p(u=u_i\mid\tilde{\mathbf{y}}_t)$, and identifying the missing pieces for computing this posterior. Using Bayes' rule, we can write\begin{IEEEeqnarray}{ll}
    p(u=u_i\mid\tilde{\mathbf{y}}_t)&=\frac{p(u=u_i)f(\tilde{\mathbf{y}}_t\mid u_i)}{f(\tilde{\mathbf{y}}_t)}\overset{(a)}{\propto} f(\tilde{\mathbf{y}}_t\mid u_i),\IEEEeqnarraynumspace\label{eq:upostjoint}
\end{IEEEeqnarray}where $(a)$ follows from two facts: i. we rely on discrete \emph{uniform} (prior) distribution on $u$, ii. $f(\tilde{\mathbf{y}}_t)$ does not depend on $u$. Consequently, the goal is to be able to compute $f(\tilde{\mathbf{y}}_t\mid u_i)$ as accurately possible. We explore this direction next.

Using marginalization and Bayes' rule we get\begin{IEEEeqnarray}{ll}
    f(\tilde{\mathbf{y}}_t\mid u_i)&=\int f(\tilde{\mathbf{y}}_t,\alpha\mid u_i) d\alpha\nonumber\\
    &=\int f(\alpha\mid u_i) f(\tilde{\mathbf{y}}_t\mid \alpha,u_i)d\alpha.
\end{IEEEeqnarray}The integrand consists of two factors, of which, the second factor, $f(\tilde{\mathbf{y}}_t\mid \alpha,u_i)$, is well-defined, provided the distribution of noise. However, the first factor, $f(\alpha\mid u_i)$, is unavailable.

In the absence of this knowledge, we approximate the unknown distribution with the best proxy available at-hand that uses all the available measurements i.e., we replace $f(\alpha\mid u_i)$ with $f(\alpha_i\mid\tilde{\mathbf{y}}_t,u_i;\hat{\gamma}_{i,t})$. The required quantity, $f(\alpha_i\mid\tilde{\mathbf{y}}_t,u_i;\hat{\gamma}_{i,t})$, was computed in the previous subsection, and we update this proxy as we collect more measurements. In other words, we compute the following \emph{approximate} likelihood\begin{IEEEeqnarray}{ll}
    \hat{f}_t(\tilde{\mathbf{y}}_t\mid u_i)
    &=\int f(\alpha_i\mid\tilde{\mathbf{y}}_t,u_i;\hat{\gamma}_{i,t}) f(\tilde{\mathbf{y}}_t\mid \alpha_i,u_i)d\alpha_i\nonumber\\
    &=\frac{1}{\pi^{(t+1)N_v}\mathrm{det}(\tilde{\bm{\Sigma}}_{i,t})}\exp{\left(-\tilde{\mathbf{e}}_{i,t}^H\tilde{\bm{\Sigma}}^{-1}_{i,t}\tilde{\mathbf{e}}_{i,t}\right)},\label{eq:likelyu_i}\IEEEeqnarraynumspace
\end{IEEEeqnarray}where $\tilde{\mathbf{e}}_{i,t}=\tilde{\mathbf{y}}_t-\tilde{\bm{\mu}}_{i,t}$, $\tilde{\bm{\mu}}_{i,t}=\sqrt{P_s}\hat{\mu}_{\alpha_i,t}\left(\tilde{\bm{\beta}}_t(u_i)\otimes\bm{\phi}_{N_v}(u_i)\right)$ and $\tilde{\bm{\Sigma}}_{i,t}=P_s\hat{\sigma}^2_{\alpha_i,t}\left(\tilde{\bm{\beta}}_t(u_i)\tilde{\bm{\beta}}_t(u_i)^H\right)\otimes\left(\bm{\phi}_{N_v}(u_i)\bm{\phi}_{N_v}(u_i)^H\right)$ $+\sigma_n^2\mathbf{I}$. $(\hat{\mu}_{\alpha_i,t},\hat{\sigma}^2_{\alpha_i,t})$ were estimated in the previous subsection. 

Next, we demonstrate the procedure to compute the required quantities in (\ref{eq:likelyu_i}) efficiently. We provide the final result below\begin{IEEEeqnarray}{ll}
    \mathrm{det}(\tilde{\bm{\Sigma}}_{i,t})&=(P_s\sigma^2_{\alpha_i,t}g_t(u_i)\lVert\bm{\phi}_{N_v}(u_i)\rVert_2^2+\sigma_n^2)\nonumber\\
    &\quad\times(\sigma_n^2)^{(t+1)N_v-1}\\
    \tilde{\mathbf{e}}_{i,t}^H\tilde{\bm{\Sigma}}_{i,t}^{-1}\tilde{\mathbf{e}}_{i,t}&=\sigma_n^{-2}\lVert\tilde{\mathbf{e}}_{i,t}\rVert_2^2-\frac{P_s\sigma^2_{\alpha_i,t}}{\sigma_n^{2}}\nonumber\\
    &\quad\times\frac{\vert\bm{\phi}^H_{N_v}(u_i)\sum_{t'=0}^t\beta^c_{t'}(u_i)\tilde{\mathbf{e}}_{i,t'}\vert^2}{(P_s\sigma^2_{\alpha_i,t}g_t(u_i)\lVert\bm{\phi}_{N_v}(u_i)\rVert_2^2+\sigma_n^2)}.
\end{IEEEeqnarray}Once the required quantities in (\ref{eq:likelyu_i}) are computed, we compute the likelihood estimate, $\hat{f}_t(\tilde{\mathbf{y}}_t\mid u_i)$, for all grid points $u_i,i\in [G]$, in the RoI. We get the estimate for the posterior as\begin{equation}
    \hat{p}_t(u=u_i\mid\tilde{\mathbf{y}}_t)=\frac{\hat{f}_t(\tilde{\mathbf{y}}_t\mid u_i)}{\sum_{i'}\hat{f}_t(\tilde{\mathbf{y}}_t\mid u_{i'})}.\label{eq:postu_i}
\end{equation}Let $\hat{\mathbf{p}}_t=[\hat{p}_t(u_0\mid\tilde{\mathbf{y}}_t),\hat{p}_t(u_1\mid\tilde{\mathbf{y}}_t),\ldots,\hat{p}_t(u_{G-1}\mid\tilde{\mathbf{y}}_t)]^T$. We adapt the SVAM beamformer, $\mathbf{f}_{t+1}$, for the next snapshot based on the computed posterior in (\ref{eq:postu_i}). We discuss this step next.
\subsection{Proposed Adaptive SVAM Beamforming
Algorithm}
The aim is to ensure that the DoA lies within the passband of the SVAM beamformer, $\mathbf{f}_t$, so that effective received SNR for inference is high. 
A high effective SNR helps further to ensure a successful beam alignment. The presented approach still may not ensure that the DoA always stays in the passband. We circumvent this issue by evaluating the posterior $\hat{p}_t(u_i\mid\tilde{\mathbf{y}}_t)$ over the entire RoI. This allows the posterior mass to move freely and helps to recover
when the SVAM beamformer adapted prematurely in the incorrect region. In other words, even if we adapt the SVAM beamformer, the RoI is kept fixed.
\begin{algorithm}[tp]
	\SetAlgoLined
	\KwResult{Beam direction: $u^*$}
	\KwIn{$L,N_v,P_s,\sigma_n^2,\mathrm{beam\_dir},\mathrm{BW}_{\mathbf{f}}:=\mathrm{BW}_{\mathrm{initial}},p_{\mathrm{thresh}}$}
	Initialize: $\mathrm{beam\_spec\_initial}:=\{\mathrm{beam\_dir,\mathrm{BW}_{\mathbf{f}}}\},\mathbf{f}_0:=\mathrm{beamformer\_design}(\mathrm{beam\_spec\_initial})$\\
	\For{$l:=0\>\mathrm{to\>}L-1$}{
	$\mathbf{y}_l=\mathbf{w}_l^H\mathbf{x}_l$\qquad{\it(New measurement, $\mathbf{w}_l$ as in (\ref{eq:wlstructure}))}\\
 \If{$\mathrm{mod}(l+1,N_v)=0$}{
 $t:=(l+1)/N_v-1$\\
 Form $\mathbf{y}_t$, then $\tilde{\mathbf{y}}_t$ as in (\ref{eq:virtualmeas}) and (\ref{eq:virtualmeasstack}) respectively\\
 {\it($\alpha$ Prior)}: Compute $\hat{\gamma}_{i,t}$ as in (\ref{eq:gammaupdate}), $i\in[G]$\\
 {\it($\alpha$ Posterior)}: Compute $\hat{\mu}_{\alpha_i,t},\hat{\sigma}^2_{\alpha_i,t}$ as in (\ref{eq:alphaposterior})\\
 {\it(Likelihood $\mid u_i$)}: Compute likelihood in (\ref{eq:likelyu_i})\\
 {\it($u_i$ Posterior)}: $\hat{\mathbf{p}}_t:=$Compute pmf in (\ref{eq:postu_i})\\
 $\mathrm{BW}_{\mathrm{check}}:=0.5\times\mathrm{BW}_{\mathbf{f}}$\\
 $[\mathrm{peak\_prob},\mathrm{beam\_spec}]:=\mathrm{cumul\_peak}(\hat{\mathbf{p}}_t,\mathrm{BW}_{\mathrm{check}})$\hfill{\it(Algorithm~\ref{alg:cumulpeak})}\\
 \While{$\mathrm{peak\_prob}< p_{\mathrm{thresh}}$}{
 $\mathrm{BW}_{\mathrm{check}}:=2\times\mathrm{BW}_{\mathrm{check}}$\\
 $[\mathrm{peak\_prob},\mathrm{beam\_spec}]:=\mathrm{cumul\_peak}(\hat{\mathbf{p}}_t,\mathrm{BW}_{\mathrm{check}})$\hfill{\it(Algorithm~\ref{alg:cumulpeak})}
 }
 $\mathrm{BW}_{\mathbf{f}}:=\mathrm{BW}_{\mathrm{check}}$\\
 $\mathbf{f}_{t+1}$:=$\mathrm{beamformer\_design}(\mathrm{beam\_spec})$
 
 
 }
 }
 $i^*:=\arg\max_i\>\hat{p}_{L/N_v-1}(u_i\mid\tilde{\mathbf{y}}_{L/N_v-1}),u^*:=u_{\mathrm{i^*}}$
 
	\caption{Proposed Adaptive SVAM Beamforming}\label{alg:adaptbeamalign}
\end{algorithm}
\subsubsection{Proposed Algorithm}
We initially set the beam direction of the SVAM beamformer, $\mathbf{f}_0$, to the centre of the RoI, and the beamwidth to $\mathrm{BW}_{\mathbf{f}}=\mathrm{BW}_{\mathrm{initial}}<2$ in $u$-space; $\mathrm{BW}_{\mathrm{initial}}$ chosen to cover the RoI. We adapt the SVAM beamformer only if there is a sufficient posterior mass \emph{concentrated} around the posterior mode. More specifically, we adapt the beamformer if the peak of the posterior mass in any contiguous span of a specific beamwidth, that includes the posterior mode, exceeds a fixed threshold $p_{\mathrm{thresh}}$. We begin with setting the beamwidth for this search to $\mathrm{BW}_{\mathrm{check}}=0.5\times \mathrm{BW}_{\mathbf{f}}$ i.e., half of the current beamwidth for the SVAM sensing. If the threshold condition is not satisfied, $\mathrm{BW}_{\mathrm{check}}$ is doubled. This continues until the posterior threshold condition is met. Note that in the default scenario, the next beam resorts to the initial beam specifications. When the posterior threshold condition is met, the corresponding contiguous span of angular grid points is selected, and the next SVAM beamformer, $\mathbf{f}_{t+1}$, is designed to cover the selected region. The threshold parameter $p_{\mathrm{thresh}}$ may be set to a fixed value or chosen dynamically. A lower value for $p_{\mathrm{thresh}}$ allows the SVAM beamformer to adapt often, whereas a higher value makes the adaptations more cautious. We discuss the role of this crucial parameter in detail in Section~\ref{sec:numsecalphaunknown}. Algorithm~\ref{alg:adaptbeamalign} summarizes the overall approach. 
\begin{remark}The proposed adaptive beam search
mechanism has the ability to both refine, as well as correct erroneous adaptations. In the case when the beam is adapted in the wrong portion of the spatial region, the posterior mass along with the mode is expected to either shift or widen as more measurements are collected. The beam gets rectified within the proposed scheme as the search is carried out around the updated posterior mode, and it has the ability to widen the beamwidth beyond the $\mathrm{BW}_{\mathbf{f}}$ presently in use for the sensing.\end{remark} The presented approach to adapt the beamformer has an equivalent representation within the \emph{compact} hierarchical codebook, and is discussed next.

\begin{algorithm}[tp]
	\SetAlgoLined
	\KwResult{$\mathrm{peak\_prob},\mathrm{beam\_spec}:=\{\mathrm{BD\_sel},\mathrm{BW}'\}$}
	\KwIn{$\mathbf{p},\mathrm{BW}'$}
 Initialize: $\mathbf{1}_{\mathrm{BW}'}$ (vector of $1$'s of size $\equiv\mathrm{BW}'$)\\
 $\mathrm{mode}:=\arg\max_j\>\mathbf{p}(j)$\\
 $\mathrm{cumul\_prob}:=\mathbf{p}*\mathbf{1}_{\mathrm{BW}'}$\hfill(`$*$': convolution operation)\\
 $[\mathrm{peak\_prob},k]:=\max(\mathrm{cumul\_prob}(\mathrm{mode}:\mathrm{mode}+\mathrm{size}(\mathbf{1}_{\mathrm{BW}'})-1)$\\
 $\mathrm{BD\_sel}:=\mathbf{u}_{\mathrm{grid}}(\mathrm{mode}+k)-\mathrm{BW}'/2$	\caption{Cumulative Posterior Peak}\label{alg:cumulpeak}
\end{algorithm}
\begin{algorithm}[tp]
	\SetAlgoLined
	\KwResult{$l^h_{\mathrm{final}}\in[\log_2G],k^h_{\mathrm{final}}\in[2^{l^h_{\mathrm{final}}}-1]$}
	\KwIn{$l^h_{\mathrm{init}}$ {\it(current level in hierarchy)}, {\it$\mathbf{p}$ (posterior pmf)}, $G$, $p_{\mathrm{thresh}}$}
 Initialize: $l^h:=l^h_{\mathrm{init}}+1,G_{l^h}:=2^{l^h}$\\
 $\mathrm{mode}:=\arg\max_j\>\mathbf{p}(j)$\\
 {\it(node at level $l^h$)}: $k^h:=\mathrm{floor}\left(\mathrm{mode}\times\frac{G_{l^h}}{G}\right)$\\
 \For{$\bar{l}^h:=l^h\>\mathrm{to}\>0$\quad{\it(descending)}}{
 $\mathrm{node\_k\_ind}:=\frac{G}{G_{l^h}}k^h:\frac{G}{G_{l^h}}(k^h+1)-1$\\
 \eIf{$\mathrm{sum}(\mathbf{p}(\mathrm{node\_k\_ind}))\geq p_{\mathrm{thresh}}$}{
 {\it break}
 }{
 $k^h:=\mathrm{floor}(k^h/2)$\\
 $G_{l^h}:=G_{l^h}/2$
 }
 }
 $l^h_{\mathrm{final}}:=l^h,k^h_{\mathrm{final}}:=k^h$	\caption{Beam Search In Hierarchical Codebook}\label{alg:beamhiercodebk}
\end{algorithm}

\subsubsection{Using Hierarchical Codebook}The proposed scheme in Algorithm~\ref{alg:adaptbeamalign} adopts a flexible beam design which requires beam direction and beamwidth to design the SVAM beamformer. In some cases, such flexible beam designs may not be possible due to complexity, and a smaller codebook may be desired. We take an example of the (binary) hierarchical codebook\cite{alkhateeb14} and demonstrate the procedure to select a codeword based on the posterior computed in (\ref{eq:postu_i}). The procedure closely mimics the approach taken in Algorithm~\ref{alg:adaptbeamalign}, while constraining the beam to belong to the hierarchical codebook. 

The basic idea is to begin at one level below in hierarchy compared to the current level used for SVAM sensing, and traverse \emph{up} the hierarchy to satisfy the posterior condition. At one-level below current, we begin with that node which contains the posterior mode. The principle in doing so, is to try and ensure that the DoA is included in the next beam. We summarize this beam search in Algorithm~\ref{alg:beamhiercodebk} which can replace line $11-18$ in Algorithm~\ref{alg:adaptbeamalign}. We also study the impact of using a hierarchical codebook on the performance of the proposed adaptive scheme in simulation Section~\ref{sec:rmsevssnrsmallcodebook}.\begin{remark}
    A variation on the beam search Algorithm~\ref{alg:beamhiercodebk} is to initialize the search at a deeper hierarchical level, which can be favourable to converge early at high SNRs, but may lead to premature beam focusing at low SNRs. Similar idea can be incorporated with Algorithm~\ref{alg:adaptbeamalign} as well. This is not explored further in this paper, but left as future work.  
\end{remark}
\subsubsection{On HiePM-based Beam Alignment}
The posterior computed in (\ref{eq:likelyu_i}) involves approximating the posterior on $\alpha$, $f(\alpha\mid u_i)$, using a Gaussian density, $f(\alpha_i\mid\tilde{\mathbf{y}}_t,u_i;\hat{\gamma}_{i,t})$. This allows to bypass a grid approach on $\alpha$, which was explored in \cite{ronquillo19} under Algorithm 1. Although the grid approach  can be more accurate, it is computationally expensive. Also initializing an appropriate grid on $\alpha$ is a non-trivial problem.

Although the adopted prior on $\alpha$ in this work and within Alg.~2 in \cite{ronquillo19}, both are Gaussian distributed, there is an important distinction. The presented approach utilizes a \emph{parameterized} Gaussian prior unlike the fixed Gaussian prior in \cite{ronquillo19}. This has a couple of implications: a) We allow the framework to \emph{learn} the relevant parameters directly from the measurements without relying on any additional information such as prior mean or variance, b) The choice of prior here allows the framework to fit the best Gaussian posterior on $\alpha$ using the measurements, in contrast to the fixed prior case where only the posterior mean is learned from measurements and the variance in the passband depends only on the prior and the noise variance level. It is known that the inference is less sensitive to the values of higher-level hyperparameters than the values for the $\alpha$-prior distribution (see \cite{giri16} and references therein). 
This leads to more \emph{robust} learning.

Most importantly, the proposed approach processes measurements \emph{jointly} for computing the posterior on $u$ (see (\ref{eq:upostjoint})). In contrast, in \cite{ronquillo19}
the $u$-posterior is updated \emph{sequentially}. Since $\alpha$ is considered static in time, both in this work and within Algorithm 2 in \cite{ronquillo19}, and because the posteriors are approximate, it was observed that 
processing measurements \emph{jointly} instead of \emph{sequentially} performs better. In other words, the $u$-posterior at previous snapshot is not a sufficient statistic when computed \emph{approximately}. The impact is studied in 
Section~\ref{sec:rmsevssnr}.
\section{Numerical Results}\label{sec:numsecalphaunknown}
In this section we perform numerical experiments to study the performance of the proposed sensing and beam alignment Algorithm~\ref{alg:adaptbeamalign} under different scenarios, and the impact of the parameters involved. The parameters that are central to the performance include i. Size of the virtual ULA, $N_v$, ii. noise variance parameter $\sigma_n^2$, iii. posterior threshold, $p_{\mathrm{thresh}}$. In Section~\ref{sec:numsecalphaknown}, we studied the role of $N_v$; we focus on the parameters $\sigma_n^2$ and $p_{\mathrm{thresh}}$ in this section. We assume that the DoA lies in between $[0,1)$ i.e., RoI is $\frac{1}{2}$ of the entire space. The information about the RoI is incorporated into the posterior calculation. Note that the RoI can be arbitrary, and one may choose a wider or narrower RoI depending on the use case specifications. The virtual ULA constructed here has the same inter-element spacing as the physical antenna. Depending on the RoI, a wider virtual inter-element spacing can be realized\footnote{For e.g., here a $\lambda$-spacing instead of $\lambda/2$ may be realized without causing ambiguity in angular estimation, since the DoA lies in $\frac{1}{2}$ of the spatial region.}, but in the following simulation we do not exploit it further.
We set the following parameters, unless otherwise specified as: Physical antenna size: $N=64$, number of grid points uniformly spaced in RoI: $G=64$, virtual ULA size: $N_v=4$, total number of snapshots: $L=120$, number of random realizations for averaging: $Q=100$. 
The SVAM beamformer is designed using the Parks-McClellan algorithm\cite{oppenheim09} with the passband edge set to realize the desired beamwidth, and the transition width set as a fraction of the desired beamwidth.\vspace{-0.5em}
\subsection{Performance of Algorithm~\ref{alg:adaptbeamalign} as a function of SNR}\label{sec:rmsevssnr}
\begin{figure}\vspace{-1em}
    \centering
    \begin{tabular}{@{\hskip -0.2em}c}\includegraphics[width=\linewidth]{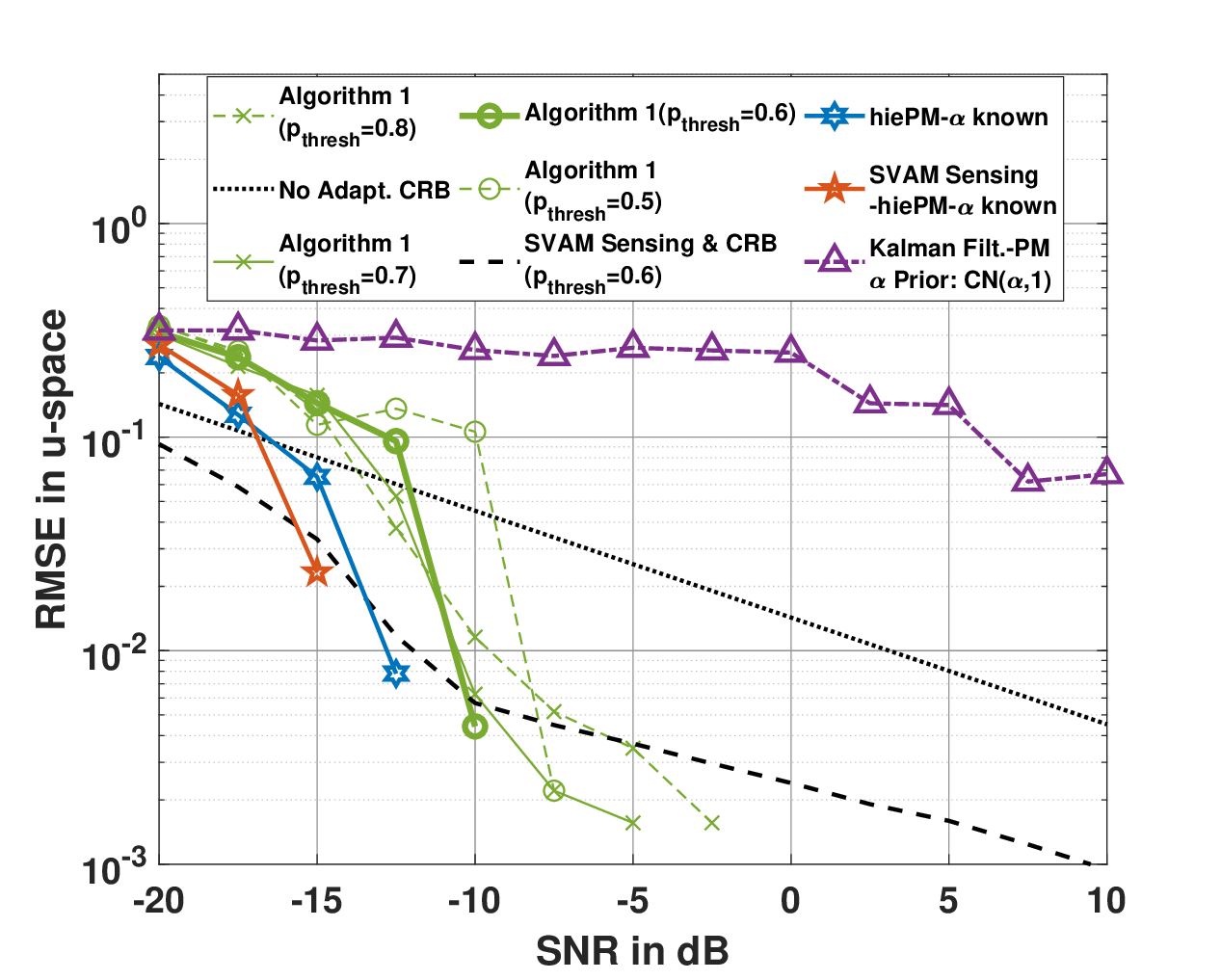}\end{tabular}\vspace{-1em}
    \caption{RMSE as a function of SNR. Performance of Algorithm 1 is observed to be close to the performance of algorithms using true $\alpha$.} 
    \label{fig:rmsevssnr}
\end{figure}In Fig.~\ref{fig:rmsevssnr}, we study the RMSE in $u$-space as a function of SNR in dB.
The red curve (with pentagram markers) plots the performance of hiePM with the SVAM sensing scheme (Section~\ref{sec:alphaknownsvamhiePM}). It is seen to improve over the hiePM algorithm in blue curve with hexagram markers, at SNR=$\{-15,-12.5\}$ dB. These two curves assume that $\alpha$ is known. At high SNR, the schemes incur zero error in this experiment as the DoA is on-grid. The green curves plots the performance of the proposed approach when $\alpha$ is unknown. The different curves correspond to different posterior threshold, $p_{\mathrm{thresh}}$, levels for adapting the beamformer. At lower thresholds, the algorithm has the ability to adapt often. This may lead to slightly unstable performance, especially at low SNR, which is observed in the green-dashed curve with $\circ$ markers where $p_{\mathrm{thresh}}=0.5$ is used.
Setting a high $p_{\mathrm{thresh}}$ requires a higher posterior mass to be accumulated for the beamformer to adapt. This adds more stability to the algorithm, but may incur slow beamformer adaptations, especially at high SNR. This is seen in green-dashed curve with $\mathrm{x}$ markers where $p_{\mathrm{thresh}}=0.8$ is set. Note that all the green curves are already close to the blue and red curves, without assuming any knowledge of $\alpha$. The curve for the Algorithm 2 in \cite{ronquillo19} is plotted in purple with $\Delta$ markers. The algorithm is provided with the true value of $\alpha$ as mean along with variance set to $1$ for the Gaussian $\alpha$-prior. It is competitive at low SNR as it hinges on the already good prior information, whereas at high SNR it disregards the provided prior. Another important reason behind the large gap between this curve, and the blue curve with hexagram markers is the fact that the static $\alpha$-case considered in \cite{ronquillo19} is not exploited\footnote{Posterior on $\alpha$ at current snapshot is used to compute likelihood of current measurement, but not utilized to update likelihoods of previous measurements.}. The proposed beam alignment procedure exploits the same, and thus demonstrates a drastic improvement as the resulting posterior on angle is more accurate. We plot the (conditional) CRB on variance of angular estimation for two schemes i. using a non-adaptive SVAM beamformer with beam direction as $0.5$ and beamwidth as $1$ ii. using adaptive SVAM beamformers chosen during runtime of Algorithm~\ref{alg:adaptbeamalign} with $p_{\mathrm{thresh}}=0.6$. As seen in Fig.~\ref{fig:rmsevssnr}, the curves using proposed approach attains (\& surpasses) the CRBs as SNR increases. Note that since the DoA is on grid which is also employed by the algorithms, they are biased. Thus, CRB is not a valid lower bound. Still it provides useful insight about the performance of the proposed inference procedure.
\subsection{Beamforming gain over time using Algorithm~\ref{alg:adaptbeamalign}}
\begin{figure}
    \centering
    \begin{tabular}{@{\hskip -0.3em}c@{\hskip -1.2em}c@{\hskip -0.3em}}
    \includegraphics[width=0.55\linewidth]{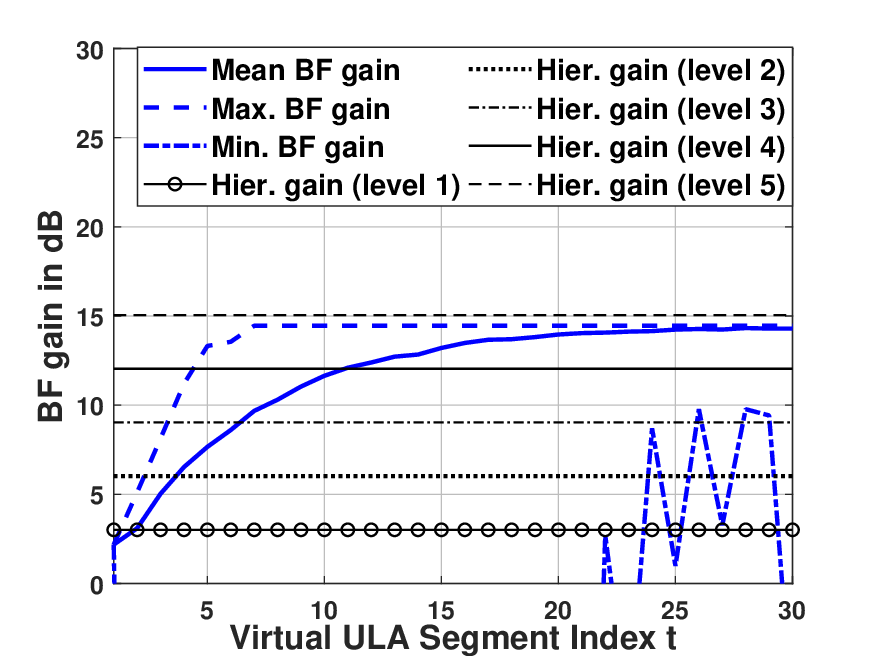}     &  \includegraphics[width=0.55\linewidth]{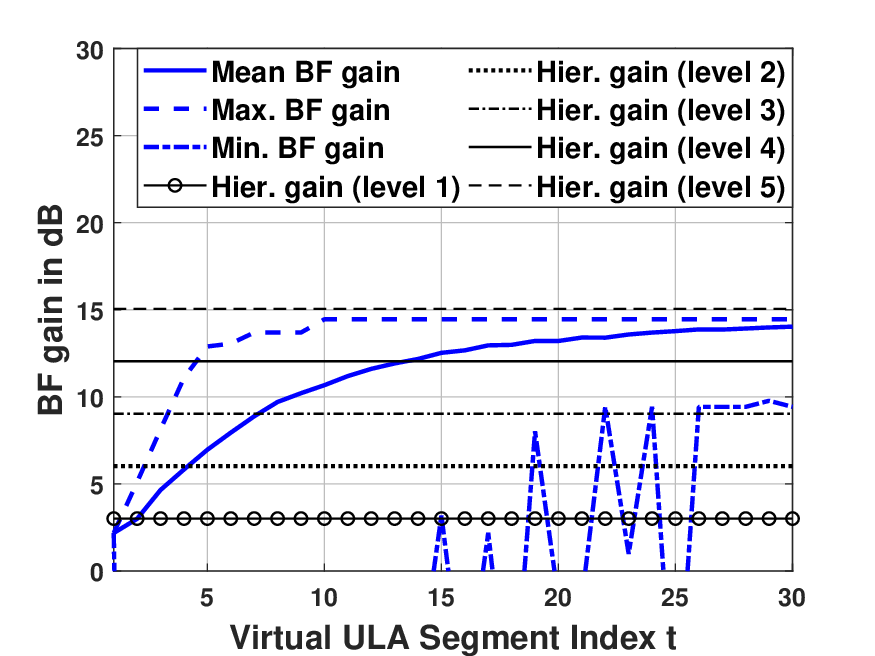}\\
     (a) $p_{\mathrm{thresh}}=0.6$   & (b) $p_{\mathrm{thresh}}=0.7$
    \end{tabular}
    \caption{Beamforming gain over time (SNR=$-10$ dB).}
    \label{fig:bfgainovertime}
\end{figure}In Fig.~\ref{fig:bfgainovertime}, we study the beamforming gain experienced by the DoA using Algorithm~\ref{alg:adaptbeamalign} at SNR of $-10$ dB. We plot the mean, maximum and the minimum beamforming gain over time (in terms of the virtual ULA segment index $t$) over $100$ random realizations. In Fig.~\ref{fig:bfgainovertime} (a) and (b) we plot the results when $p_{\mathrm{thresh}}$ is set to $0.6$ and $0.7$, respectively. As expected, at a lower posterior threshold the algorithm is more flexible, and thus the maximum beamforming gain is achieved earlier than compared to the case when a higher posterior threshold is set. In contrast, the minimum beamforming gain may be less stable when the beamformer is adapted frequently at such low SNR. This is seen in plot (a) where the minimum gain fluctuates even after $25$ virtual ULA segments i.e., $25\times 4=100$ snapshots. The solid curve representing the mean beamforming gain indicates how the gain improves over time as the beamwidth adaptively narrows. We also plot the beamforming gain achieved by an ideal beamformer. A hierarchical level $l^h\in\{1,2,\ldots,5\}$ indicates a beam focused in $2/2^{l^h+1}$ of the spatial region and consequently achieves a beamforming gain of $2^{l^h+1}/2$, in absolute scale. As observed in both the plots, beginning with a beamformer focusing on $\frac{1}{2}$ of the space (equivalently $10\log_{10}(2)\approx 3$ dB gain), the beamformer is able to adapt to $1/32$-fraction, garnering a beamforming gain of around $10\log_{10}(32)\approx 15$ dB. Thus, the proposed algorithm
adds around $12$ dB gain at $-10$ dB SNR.   
\subsection{Performance as a function of number of snapshots}\label{sec:rmsevssnap}
In Fig.~\ref{fig:rmseovertime} we plot the RMSE over time (in terms of the number of virtual ULA measurements $t$). We plot the hiePM algorithms with $\alpha$ perfectly known for comparison. In Fig.~\ref{fig:rmseovertime} (a) and (b), we plot curves corresponding to SNR$=-10$ dB and $0$ dB, respectively.
We study two curves corresponding to the proposed scheme i. $p_{\mathrm{thresh}}=0.6$ (green curve with $\mathrm{x}$ markers) ii. $p_{\mathrm{thresh}}=0.8$ (purple curve with $\diamond$ markers). The former setting allows the beamformer to aggressively adapt, which can be rewarding at high SNR, but can lead to poor beamforming gain in the low SNR regime as it adapts to incorrect regions. This is seen in the two plots, at high SNR the curve corresponding to $p_{\mathrm{thresh}}=0.6$ converges quickly, whereas at low SNR it leads to larger error and variance 
initially, compared to setting $p_{\mathrm{thresh}}=0.8$. Note that
even at low SNR and low threshold setting (Fig.~\ref{fig:rmseovertime} (a)), the yellow curve manages to adapt back to the correct spatial region, which is indicated by a drop in RMSE around $t=25$ and is able to retain the improvement over time. Overall, it is observed that a high $p_{\mathrm{thresh}}$ leads to stable yet slow convergence, which is useful when the SNR is low and if the training duration is short. At high SNR or if the training duration is longer, a low $p_{\mathrm{thresh}}$ pays off, as the \emph{algorithm is robust to recover from 
its mistakes}!

One may dynamically select $p_{\mathrm{thresh}}$ such that it is set to higher value initially during the training phase, but it may be reduced over time as the beamformer narrows further and enjoys a high beamforming gain. This is left for future work.
\subsection{Studying Impact of Noise Variance Parameter}
\begin{figure}
    \centering
    \begin{tabular}{@{\hskip -0.6em}c@{\hskip -1.2em}c}
    \includegraphics[width=0.55\linewidth]{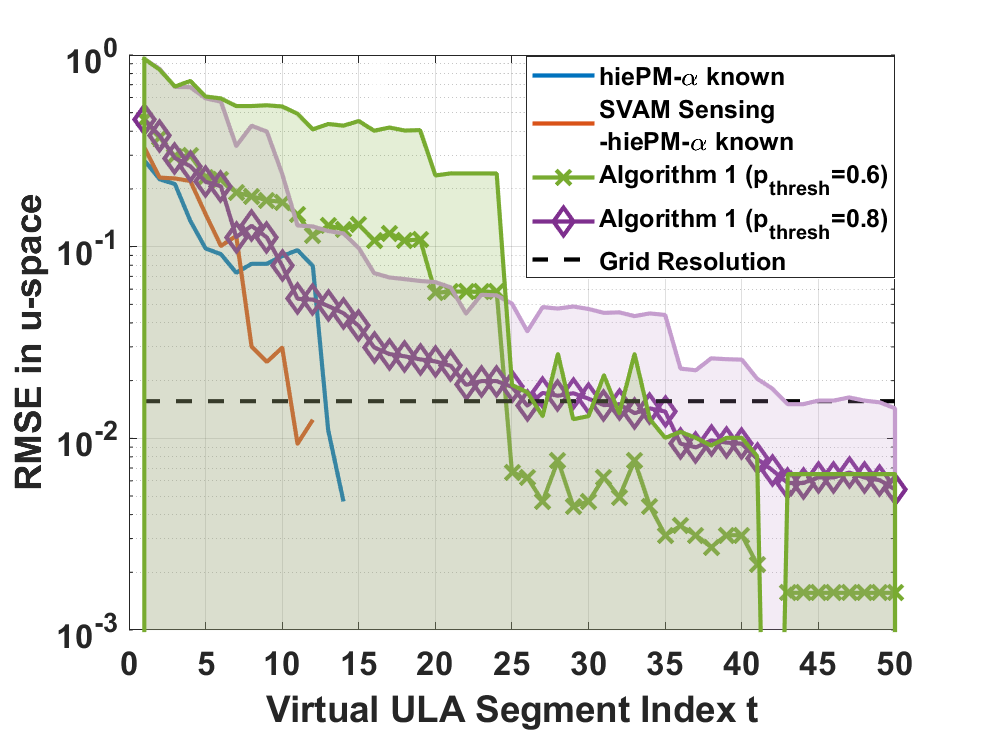}     &  \includegraphics[width=0.55\linewidth]{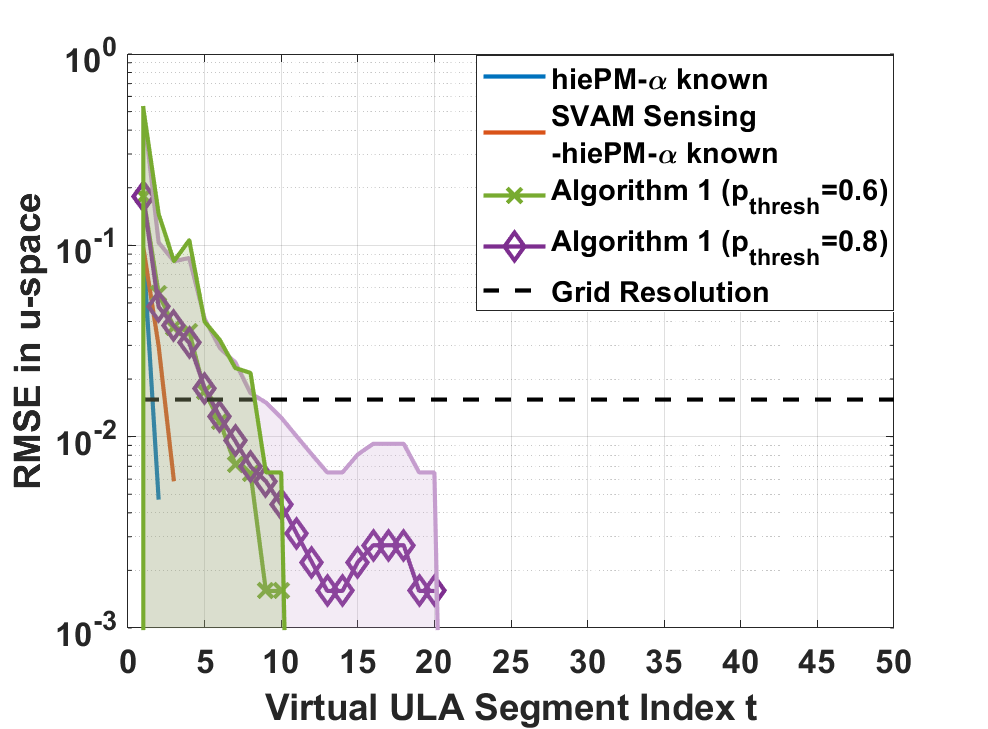}\\
     (a) SNR $=-10$ dB   & (a) SNR $=0$ dB
    \end{tabular}
    \caption{RMSE over time. Shaded region covers the RMSE $\pm$ root standard deviation of squared error over time at every time index along x-axis.}
    \label{fig:rmseovertime}
\end{figure}
\begin{figure}
    \centering
    \includegraphics[width=0.8\linewidth]{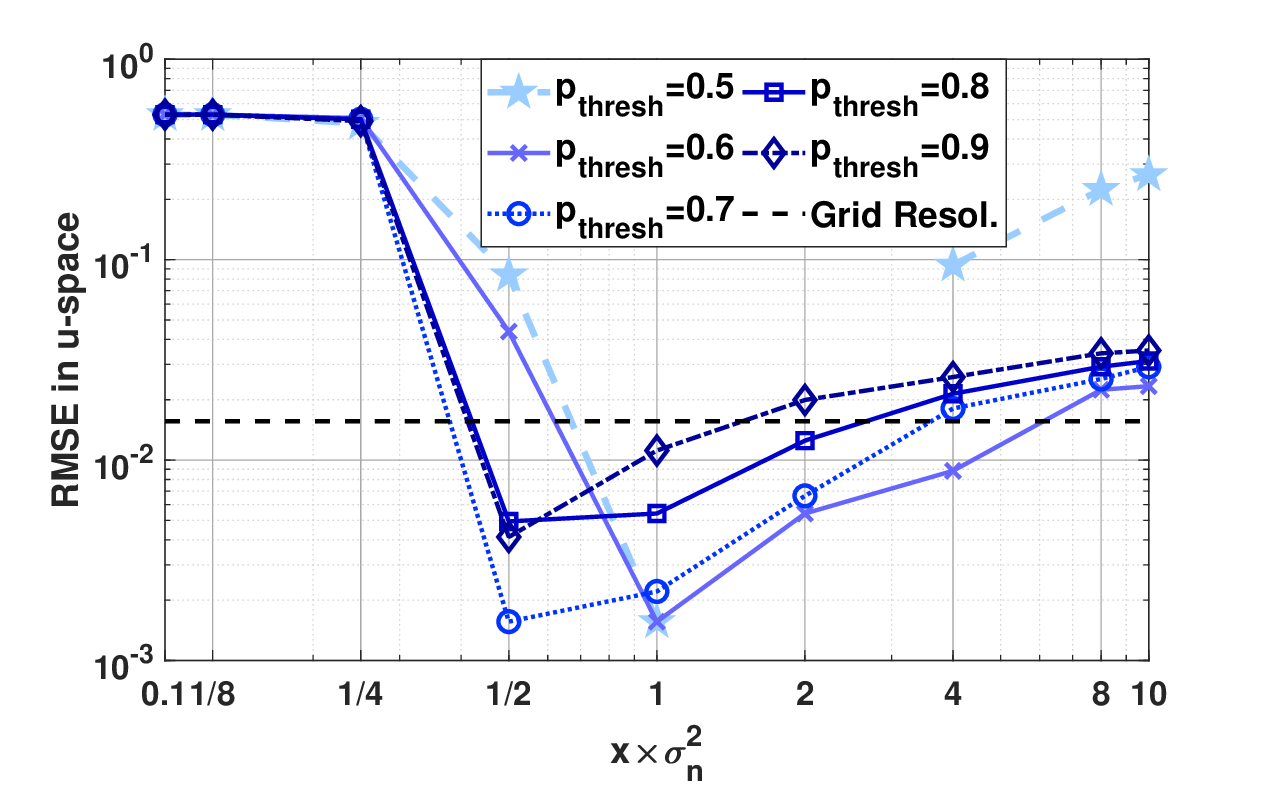}\vspace{-1em}
    \caption{RMSE vs. $x\times\sigma_n^2,x\in\{0.1,1/8,1/4,1/2,1,2,4,8,10\}$. SNR$=-10$ dB, $L=200$. As the posterior threshold increases, the optimal value for the noise variance parameter is observed to be lower than the true value.} 
    \label{fig:rmsevsnoisevarmult}
\end{figure}
We plot the RMSE as a function of different settings for the noise variance parameter in Fig.~\ref{fig:rmsevsnoisevarmult}. The optimal value is observed to be around the true noise variance value. We also plot curves corresponding to different posterior thresholds. For a low posterior threshold, the optimal value of the noise variance parameter is observed to be slightly higher ($\times 2$) than or equal to the true value. This is due to the aggressive nature of the algorithm to adapt, which can be compensated by setting a slightly higher noise variance parameter. On the other hand, a high posterior threshold already imparts more stable adaptations, and consequently does not need setting a higher noise variance parameter. In fact, at such large thresholds, the algorithm benefits from lower noise variance parameter setting, which intuitively offsets the conservative posterior threshold setting. This is observed by the curve corresponding to $p_{\mathrm{thresh}}=0.9$, where setting noise variance to $0.5\times\sigma_n^2$ leads to much better performance than setting it to the true value.\vspace{-0.5em}
\subsection{Impact of compact hierarchical codebook (Algorithm~\ref{alg:beamhiercodebk})}\label{sec:rmsevssnrsmallcodebook}
\begin{figure}
    \centering
    \includegraphics[width=0.7\linewidth]{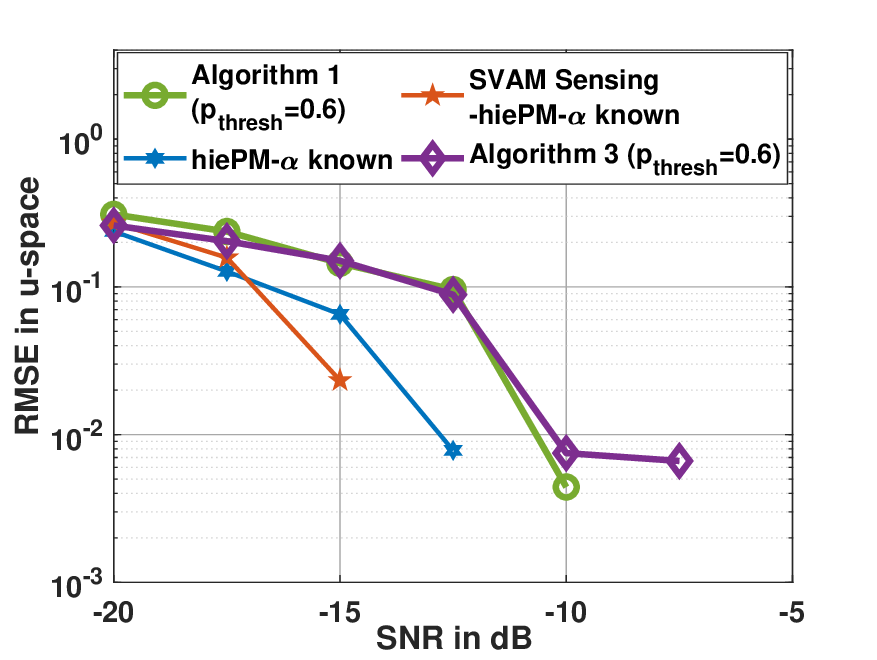}\vspace{-1em}
    \caption{RMSE as a function of SNR. A hierarchical codebook is used, which results in little to no performance loss, indicating that the overall performance gain is unaffected by the compact hierarchical codebook.}
    \label{fig:rmsevssnrsmallcodebook}
\end{figure}
In Fig.~\ref{fig:rmsevssnrsmallcodebook}, we plot the curve corresponding to using a hierarchical codebook (see purple curve with $\diamond$ markers). The approach is described in Algorithm~\ref{alg:beamhiercodebk}. We set $p_{\mathrm{thresh}}=0.6$ and compare the performance with using the flexible codebook (see green curve with $\circ$ markers). As observed in Fig.~\ref{fig:rmsevssnrsmallcodebook}, the two implementations have very similar performance. This suggests that the superior performance is due to other ingredients, namely i. improved sensing (and filter design), ii. good posterior estimate, and iii. good adaptive strategy. A large codebook size helps, but is not the key to the performance improvement demonstrated in this paper.
\section{Conclusion}\label{sec:conc}
 We proposed a novel Synthesis of Virtual Array Manifold (SVAM) sensing approach for the mmWave single RF chain systems and discussed the ensuing benefits. More specifically, the proposed sensing is demonstrated to  lead to faster and more robust beam alignment.
 We believe this contribution will have significant impact on the traditional paradigm for sensing in mmWave systems. We also proposed a novel inference scheme that estimates a posterior density on the small-scale fading coefficient and the unknown dominant path angle. Based on the proposed inference procedure, an adaptive beamforming scheme is provided that aims to collect high SNR measurements. Finally, the performance of the proposed active sensing scheme is evaluated under different scenarios, and a significant improvement over various benchmarks is demonstrated. The empirical study also reveals the impact of the different design parameters 
 on the beam alignment performance.
\section{Appendix}
 \subsection{Proof of Theorem~\ref{thm:Gnonneg}}\label{sec:Gnonnegproof}
 \begin{proof}
    We first note that the following holds\begin{equation}
        \frac{\partial}{\partial u}\bm{\phi}_M(u)=j\pi\left\{\frac{(M-1)}{2}\bm{\phi}_m(u)+\bm{\phi}^{\perp}_M(u)\right\},\label{eq:deraslincomb}
\vspace{-0.5em}    \end{equation}where $\bm{\phi}^{\perp}_M(u)=\left[\begin{array}{cccc}-\frac{(M-1)}{2}&(1-\frac{(M-1)}{2})&\ldots&\frac{(M-1)}{2}\end{array}\right]^T$ $\odot\bm{\phi}_M(u)$, `$\odot$' denotes Hadamard product. Using the result in (\ref{eq:deraslincomb}) we simplify the second term in $G$ as\vspace{-0.5em}\begin{IEEEeqnarray}{ll}
        &-\pi(N_v-1)\mathrm{Im}\left\{\left(\frac{\partial}{\partial u}\bm{\phi}_M(u)\right)^H\mathbf{F}^{\mathrm{P}}(\mathbf{F}^{\mathrm{P}})^H\bm{\phi}_M(u)\right\}\nonumber\\
        &=\pi^2(N_v-1)\mathrm{Re}\biggl\{\left(\frac{(M-1)}{2}\bm{\phi}^H_M(u)+\left(\bm{\phi}^{\perp}_M(u)\right)^H\right)\biggr.\nonumber\\
        &\quad\times\biggl.\mathbf{F}^{\mathrm{P}}(\mathbf{F}^{\mathrm{P}})^H\bm{\phi}_M(u)\biggr\}\nonumber\\
        &=\pi^2(N_v-1)\biggl\{\frac{(M-1)}{2}\bm{\phi}^H_M(u)\mathbf{F}^{\mathrm{P}}(\mathbf{F}^{\mathrm{P}})^H\bm{\phi}_M(u)\biggr.\nonumber\\
        &\biggl.\quad+\mathrm{Re}\left\{\left(\left(\bm{\phi}^{\perp}_M(u)\right)^H\right)\mathbf{F}^{\mathrm{P}}(\mathbf{F}^{\mathrm{P}})^H\bm{\phi}_M(u)\right\}\biggr\}\nonumber\\
        &\geq\pi^2(N_v-1)\biggl\{\frac{(M-1)}{2}\bm{\phi}^H_M(u)\mathbf{F}^{\mathrm{P}}(\mathbf{F}^{\mathrm{P}})^H\bm{\phi}_M(u)\biggr.\nonumber\\
        &\biggl.\quad-\left\vert\bm{\phi}^H_M(u)\mathbf{F}^{\mathrm{P}}(\mathbf{F}^{\mathrm{P}})^H\bm{\phi}^{\perp}_M(u)\right\vert\biggr\}\nonumber\\
        &=\pi^2(N_v-1)\bm{\phi}^H_M(u)\mathbf{F}^{\mathrm{P}}(\mathbf{F}^{\mathrm{P}})^H\bm{\phi}_M(u)\biggl\{\frac{(M-1)}{2}\biggr.\nonumber\\
        &\biggl.\quad-\frac{\left\vert\bm{\phi}^H_M(u)\mathbf{F}^{\mathrm{P}}(\mathbf{F}^{\mathrm{P}})^H\bm{\phi}^{\perp}_M(u)\right\vert}{\bm{\phi}^H_M(u)\mathbf{F}^{\mathrm{P}}(\mathbf{F}^{\mathrm{P}})^H\bm{\phi}_M(u)}\biggr\}.
    \end{IEEEeqnarray}Thus to ensure\vspace{-0.5em}\begin{IEEEeqnarray}{lll}
        &G&\geq 0\nonumber\\
        \implies&c\biggl\{\frac{(2N_v-1)}{6}+\frac{(M-1)}{2}\biggr.&\geq 0\nonumber\\
        &\biggl.-\frac{\left\vert\bm{\phi}^H_M(u)\mathbf{F}^{\mathrm{P}}(\mathbf{F}^{\mathrm{P}})^H\bm{\phi}^{\perp}_M(u)\right\vert}{\bm{\phi}^H_M(u)\mathbf{F}^{\mathrm{P}}(\mathbf{F}^{\mathrm{P}})^H\bm{\phi}_M(u)}\biggr\}&
    \end{IEEEeqnarray}where $c=\pi^2(N_v-1)\bm{\phi}^H_M(u)\mathbf{F}^{\mathrm{P}}(\mathbf{F}^{\mathrm{P}})^H\bm{\phi}_M(u)$. This implies\begin{IEEEeqnarray}{ll}
    \frac{\left\vert\bm{\phi}^H_M(u)\mathbf{F}^{\mathrm{P}}(\mathbf{F}^{\mathrm{P}})^H\bm{\phi}^{\perp}_M(u)\right\vert}{\bm{\phi}^H_M(u)\mathbf{F}^{\mathrm{P}}(\mathbf{F}^{\mathrm{P}})^H\bm{\phi}_M(u)}&\leq\biggl\{\frac{(3M+2N_v-4)}{6}\biggr\}\nonumber\\
    \frac{\left\vert\bm{\phi}^H_M(u)\mathbf{F}^{\mathrm{P}}(\mathbf{F}^{\mathrm{P}})^H\bm{\phi}^{\perp}_M(u)\frac{\lVert\bm{\phi}_M(u)\rVert}{\lVert\bm{\phi}^{\perp}_M(u)\rVert}\right\vert}{\bm{\phi}^H_M(u)\mathbf{F}^{\mathrm{P}}(\mathbf{F}^{\mathrm{P}})^H\bm{\phi}_M(u)}&\leq C(N,N_v),\label{eq:Gmaincond}
    \end{IEEEeqnarray}where $C(N,N_v)=\frac{\lVert\bm{\phi}_M(u)\rVert}{\lVert\bm{\phi}^{\perp}_M(u)\rVert}\biggl\{\frac{(3M+2N_v-4)}{6}\biggr\},\lVert\bm{\phi}_M(u)\rVert=\sqrt{M},\lVert\bm{\phi}^{\perp}_M(u)\rVert=\frac{\sqrt{M(M^2-1)}}{2\sqrt{3}}$. Note that $C(N,N_v)\geq \sqrt{3}$ when $N_v>1$. Let $\bm{\phi}^{\perp}_{M,\mathrm{n}}(u)=\bm{\phi}^{\perp}_M(u)\frac{\lVert\bm{\phi}_M(u)\rVert}{\lVert\bm{\phi}^{\perp}_M(u)\rVert}$. The LHS can be further simplified to get\begin{IEEEeqnarray}{ll}
        &\frac{\left\vert\bm{\phi}^H_M(u)\mathbf{F}^{\mathrm{P}}(\mathbf{F}^{\mathrm{P}})^H\bm{\phi}^{\perp}_M(u)\frac{\lVert\bm{\phi}_M(u)\rVert}{\lVert\bm{\phi}^{\perp}_M(u)\rVert}\right\vert}{\bm{\phi}^H_M(u)\mathbf{F}^{\mathrm{P}}(\mathbf{F}^{\mathrm{P}})^H\bm{\phi}_M(u)}\nonumber\\
        &=\frac{\sqrt{\bm{\phi}^H_M(u)\mathbf{F}^{\mathrm{P}}(\mathbf{F}^{\mathrm{P}})^H\bm{\phi}^{\perp}_{M,\mathrm{n}}(u)(\bm{\phi}^{\perp}_{M,\mathrm{n}}(u))^H\mathbf{F}^{\mathrm{P}}(\mathbf{F}^{\mathrm{P}})^H\bm{\phi}_M(u)}}{\bm{\phi}^H_M(u)\mathbf{F}^{\mathrm{P}}(\mathbf{F}^{\mathrm{P}})^H\bm{\phi}_M(u)}\nonumber\\
        &\overset{(a)}{=}\sqrt{\frac{\lVert\bm{\phi}_M(u)\rVert^2\bm{\phi}^H_M(u)\mathbf{F}^{\mathrm{P}}(\mathbf{F}^{\mathrm{P}})^H\mathbf{P}_{u,\perp}\mathbf{F}^{\mathrm{P}}(\mathbf{F}^{\mathrm{P}})^H\bm{\phi}_M(u)}{(\bm{\phi}^H_M(u)\mathbf{F}^{\mathrm{P}}(\mathbf{F}^{\mathrm{P}})^H\bm{\phi}_M(u))^2}-1}\nonumber\\
        &\overset{(b)}{\leq}\sqrt{\frac{\lVert\bm{\phi}_M(u)\rVert^2\lambda_{\mathrm{max}}\left((\mathbf{F}^{\mathrm{P}})^H\mathbf{P}_{u,\perp}\mathbf{F}^{\mathrm{P}}\right)}{\bm{\phi}^H_M(u)\mathbf{F}^{\mathrm{P}}(\mathbf{F}^{\mathrm{P}})^H\bm{\phi}_M(u)}-1}.
    \end{IEEEeqnarray}where in $(a)$ we express $\frac{\bm{\phi}^{\perp}_{M,\mathrm{n}}(u)(\bm{\phi}^{\perp}_{M,\mathrm{n}}(u))^H}{\lVert\bm{\phi}^{\perp}_{M,\mathrm{n}}(u)\rVert^2}=\mathbf{P}_{u,\perp}-\frac{\bm{\phi}_{M}(u)\bm{\phi}_{M}(u)^H}{\lVert\bm{\phi}_{M}(u)\rVert^2}$. $\mathbf{P}_{u,\perp}=[\bm{\phi}_{M}(u)\>\> \bm{\phi}^{\perp}_{M}(u)]$ $\times\left[\begin{array}{cc}
       \lVert\bm{\phi}_{M}(u)\rVert^2  &  0\\
       0  & \lVert\bm{\phi}^{\perp}_{M}(u)\rVert^2 
    \end{array}\right]^{-1}[\bm{\phi}_{M}(u)\>\> \bm{\phi}^{\perp}_{M}(u)]$. In $(b)$ we upper bound the quadratic term using the largest eigenvalue, $\lambda_{\mathrm{max}}\left((\mathbf{F}^{\mathrm{P}})^H\mathbf{P}_{u,\perp}\mathbf{F}^{\mathrm{P}}\right)$, of the Hermitian-symmetric matrix $(\mathbf{F}^{\mathrm{P}})^H\mathbf{P}_{u,\perp}\mathbf{F}^{\mathrm{P}}$. Thus a more stricter condition that satisfies the inequality in (\ref{eq:Gmaincond}) is given by\begin{IEEEeqnarray}{ll}
       \sqrt{\frac{\lVert\bm{\phi}_M(u)\rVert^2\lambda_{\mathrm{max}}\left((\mathbf{F}^{\mathrm{P}})^H\mathbf{P}_{u,\perp}\mathbf{F}^{\mathrm{P}}\right)}{\bm{\phi}^H_M(u)\mathbf{F}^{\mathrm{P}}(\mathbf{F}^{\mathrm{P}})^H\bm{\phi}_M(u)}-1}&\leq C(N,N_v)\IEEEeqnarraynumspace\end{IEEEeqnarray}which implies\begin{IEEEeqnarray}{ll}
       \frac{\bm{\phi}^H_M(u)\mathbf{F}^{\mathrm{P}}(\mathbf{F}^{\mathrm{P}})^H\bm{\phi}_M(u)}{\lVert\bm{\phi}_M(u)\rVert^2}\geq\frac{\lambda_{\mathrm{max}}\left((\mathbf{F}^{\mathrm{P}})^H\mathbf{P}_{u,\perp}\mathbf{F}^{\mathrm{P}}\right)}{C^2(N,N_v)+1}\IEEEeqnarraynumspace.\label{eq:Gcondstrict1}
    \end{IEEEeqnarray}Since $C^2(N,N_v)\geq 3$, an even stricter but simplified condition than (\ref{eq:Gcondstrict1}) is\begin{equation}
        \frac{\bm{\phi}^H_M(u)\mathbf{F}^{\mathrm{P}}(\mathbf{F}^{\mathrm{P}})^H\bm{\phi}_M(u)}{\lVert\bm{\phi}_M(u)\rVert^2}\geq\frac{\lambda_{\mathrm{max}}\left((\mathbf{F}^{\mathrm{P}})^H\mathbf{P}_{u,\perp}\mathbf{F}^{\mathrm{P}}\right)}{4}.\label{eq:Gcondstrict2}
    \end{equation}This concludes the proof.  
\end{proof}
 \bibliographystyle{IEEEbib}
\bibliography{main}

\begin{thebibliography}{10}

\bibitem{pi11}
Z.~{Pi} and F.~{Khan},
\newblock ``An introduction to millimeter-wave mobile broadband systems,''
\newblock {\em IEEE Communications Magazine}, vol. 49, no. 6, pp. 101--107,
  June 2011.

\bibitem{rappaport13}
T.~S. {Rappaport}, S.~{Sun}, R.~{Mayzus}, H.~{Zhao}, Y.~{Azar}, K.~{Wang},
  G.~N. {Wong}, J.~K. {Schulz}, M.~{Samimi}, and F.~{Gutierrez},
\newblock ``Millimeter wave mobile communications for 5{G} cellular: It will
  work!,''
\newblock {\em IEEE Access}, vol. 1, pp. 335--349, 2013.

\bibitem{bastug17}
E.~Bastug, M.~Bennis, M.~Medard, and M.~Debbah,
\newblock ``Toward interconnected virtual reality: Opportunities, challenges,
  and enablers,''
\newblock {\em IEEE Communications Magazine}, vol. 55, no. 6, pp. 110--117,
  2017.

\bibitem{ghosh19}
A.~Ghosh, A.~Maeder, M.~Baker, and D.~Chandramouli,
\newblock ``5g evolution: A view on 5g cellular technology beyond 3gpp release
  15,''
\newblock {\em IEEE Access}, vol. 7, pp. 127639--127651, 2019.

\bibitem{doan04}
C.H. Doan, S.~Emami, D.A. Sobel, A.M. Niknejad, and R.W. Brodersen,
\newblock ``Design considerations for 60 ghz cmos radios,''
\newblock {\em IEEE Communications Magazine}, vol. 42, no. 12, pp. 132--140,
  2004.

\bibitem{chiu19}
S.~Chiu, N.~Ronquillo, and T.~Javidi,
\newblock ``Active learning and csi acquisition for mmwave initial alignment,''
\newblock {\em IEEE Journal on Selected Areas in Communications}, vol. 37, no.
  11, pp. 2474--2489, 2019.

\bibitem{alkhateeb14}
A.~{Alkhateeb}, O.~{El Ayach}, G.~{Leus}, and R.~W. {Heath},
\newblock ``Channel estimation and hybrid precoding for millimeter wave
  cellular systems,''
\newblock {\em IEEE Journal of Selected Topics in Signal Processing}, vol. 8,
  no. 5, pp. 831--846, Oct 2014.

\bibitem{ronquillo19}
N.~Ronquillo, S.~Chiu, and T.~Javidi,
\newblock ``Sequential learning of csi for mmwave initial alignment,''
\newblock in {\em 2019 53rd Asilomar Conference on Signals, Systems, and
  Computers}, 2019, pp. 1278--1283.

\bibitem{akdim20}
N.~Akdim, C.~N. Manchón, M.~Benjillali, and P.~Duhamel,
\newblock ``Variational hierarchical posterior matching for mmwave wireless
  channels online learning,''
\newblock in {\em 2020 IEEE 21st International Workshop on Signal Processing
  Advances in Wireless Communications (SPAWC)}, 2020, pp. 1--5.

\bibitem{pote19}
R.~R. Pote and B.~D. Rao,
\newblock ``Reduced dimension beamspace design incorporating nested array for
  mmwave channel estimation,''
\newblock in {\em Asilomar Conference on Signals, Systems, and Computers},
  2019, pp. 1212--1216.

\bibitem{lin20}
Y.~Lin and T.~Yang,
\newblock ``Random sbt precoding for angle estimation of mmwave massive mimo
  systems using sparse arrays spacing,''
\newblock {\em IEEE Access}, vol. 8, pp. 163380--163393, 2020.

\bibitem{chen20}
P.~Chen and P.~P. Vaidyanathan,
\newblock ``Convolutional beamspace for linear arrays,''
\newblock {\em IEEE Trans. on Signal Proc.}, vol. 68, pp. 5395--5410, 2020.

\bibitem{silverstein91}
S.D. Silverstein, W.E. Engeler, and J.A. Tardif,
\newblock ``Parallel architectures for multirate superresolution spectrum
  analyzers,''
\newblock {\em IEEE Transactions on Circuits and Systems}, vol. 38, no. 4, pp.
  449--453, 1991.

\bibitem{tkacenko01}
A.~Tkacenko and P.P. Vaidyanathan,
\newblock ``The role of filter banks in sinusoidal frequency estimation,''
\newblock {\em Journal of the Franklin Institute}, vol. 338, no. 5, pp.
  517--547, 2001.

\bibitem{ramasamy12}
D.~Ramasamy, S.~Venkateswaran, and U.~Madhow,
\newblock ``Compressive adaptation of large steerable arrays,''
\newblock in {\em 2012 Information Theory and Applications Workshop}, 2012, pp.
  234--239.

\bibitem{berraki14}
D.~E. Berraki, S.~M.~D. Armour, and A.~R. Nix,
\newblock ``Application of compressive sensing in sparse spatial channel
  recovery for beamforming in mmwave outdoor systems,''
\newblock in {\em 2014 IEEE Wireless Communications and Networking Conference
  (WCNC)}, 2014, pp. 887--892.

\bibitem{alkhateeb15}
A.~Alkhateeb, G.~Leus, and R.~W. Heath,
\newblock ``Compressed sensing based multi-user millimeter wave systems: How
  many measurements are needed?,''
\newblock in {\em 2015 IEEE International Conference on Acoustics, Speech and
  Signal Processing (ICASSP)}, 2015, pp. 2909--2913.

\bibitem{hur13}
S.~Hur, T.~Kim, D.~J. Love, J.~V. Krogmeier, T.~A. Thomas, and A.~Ghosh,
\newblock ``Millimeter wave beamforming for wireless backhaul and access in
  small cell networks,''
\newblock {\em IEEE Transactions on Communications}, vol. 61, no. 10, pp.
  4391--4403, 2013.

\bibitem{xiao16}
Z.~Xiao, T.~He, P.~Xia, and X.~Xia,
\newblock ``Hierarchical codebook design for beamforming training in
  millimeter-wave communication,''
\newblock {\em IEEE Trans. on Wireless Communications}, vol. 15, no. 5, pp.
  3380--3392, 2016.

\bibitem{zhang17}
J.~Zhang, Y.~Huang, Q.~Shi, J.~Wang, and L.~Yang,
\newblock ``Codebook design for beam alignment in millimeter wave communication
  systems,''
\newblock {\em IEEE Transactions on Communications}, vol. 65, no. 11, pp.
  4980--4995, 2017.

\bibitem{liu17}
C.~Liu, M.~Li, S.~V. Hanly, I.~B. Collings, and P.~Whiting,
\newblock ``Millimeter wave beam alignment: Large deviations analysis and
  design insights,''
\newblock {\em IEEE Journal on Selected Areas in Communications}, vol. 35, no.
  7, pp. 1619--1631, 2017.

\bibitem{liu22}
C.~Liu, L.~Zhao, M.~Li, and L.~Yang,
\newblock ``Adaptive beam search for initial beam alignment in millimetre-wave
  communications,''
\newblock {\em IEEE Trans. on Vehicular Technology}, vol. 71, no. 6, pp.
  6801--6806, 2022.

\bibitem{sohrabi22}
F.~Sohrabi, T.~Jiang, W.~Cui, and W.~Yu,
\newblock ``Active sensing for communications by learning,''
\newblock {\em IEEE Journal on Selected Areas in Communications}, vol. 40, no.
  6, pp. 1780--1794, 2022.

\bibitem{wei23}
Y.~Wei, Z.~Zhong, and V.~Y.~F. Tan,
\newblock ``Fast beam alignment via pure exploration in multi-armed bandits,''
\newblock {\em IEEE Transactions on Wireless Communications}, vol. 22, no. 5,
  pp. 3264--3279, 2023.

\bibitem{hussain19}
M.~Hussain and N.~Michelusi,
\newblock ``Second-best beam-alignment via bayesian multi-armed bandits,''
\newblock in {\em 2019 IEEE Global Communications Conference (GLOBECOM)}, 2019,
  pp. 1--6.

\bibitem{pote23asil}
R.~R. Pote and B.~D. Rao,
\newblock ``Novel sensing methodology for initial alignment using mmwave phased
  arrays,''
\newblock in {\em Asilomar Conference on Signals, Systems, and Computers},
  2023.

\bibitem{trees02}
H.~L.~Van Trees,
\newblock ``Optimum array processing: Part {IV} of {D}etection, {E}stimation,
  and {M}odulation {T}heory,''
\newblock {\em John Wiley \& Sons, Ltd}, 2002.

\bibitem{moreira13}
Alberto Moreira, Pau Prats-Iraola, Marwan Younis, Gerhard Krieger, Irena
  Hajnsek, and Konstantinos~P. Papathanassiou,
\newblock ``A tutorial on synthetic aperture radar,''
\newblock {\em IEEE Geoscience and Remote Sensing Magazine}, vol. 1, no. 1, pp.
  6--43, 2013.

\bibitem{waldschmidt21}
C.~Waldschmidt, J.~Hasch, and W.~Menzel,
\newblock ``Automotive radar — from first efforts to future systems,''
\newblock {\em IEEE Journal of Microwaves}, vol. 1, no. 1, pp. 135--148, 2021.

\bibitem{oppenheim09}
A.~V. Oppenheim and R.~W. Schafer,
\newblock {\em Discrete-Time Signal Processing},
\newblock Prentice Hall Press, USA, 3rd edition, 2009.

\bibitem{pote23doa}
R.~R. Pote and B.~D. Rao,
\newblock ``Maximum likelihood-based gridless doa estimation using structured
  covariance matrix recovery and sbl with grid refinement,''
\newblock {\em IEEE Trans. on Signal Proc.}, vol. 71, pp. 802--815, 2023.

\bibitem{moffet68}
A.~Moffet,
\newblock ``Minimum-redundancy linear arrays,''
\newblock {\em IEEE Transactions on Antennas and Propagation}, vol. 16, no. 2,
  pp. 172--175, 1968.

\bibitem{pal10}
P.~Pal and P.~P. Vaidyanathan,
\newblock ``Nested arrays: A novel approach to array processing with enhanced
  degrees of freedom,''
\newblock {\em IEEE Transactions on Signal Processing}, vol. 58, no. 8, pp.
  4167--4181, 2010.

\bibitem{vaidyanathan11}
P.~P. Vaidyanathan and P.~Pal,
\newblock ``Sparse sensing with co-prime samplers and arrays,''
\newblock {\em IEEE Transactions on Signal Processing}, vol. 59, no. 2, pp.
  573--586, 2011.

\bibitem{va16}
Vutha Va, Haris Vikalo, and Robert~W. Heath,
\newblock ``Beam tracking for mobile millimeter wave communication systems,''
\newblock in {\em 2016 IEEE Global Conference on Signal and Information
  Processing (GlobalSIP)}, 2016, pp. 743--747.

\bibitem{song19}
X.~Song, S.~Haghighatshoar, and G.~Caire,
\newblock ``Efficient beam alignment for millimeter wave single-carrier systems
  with hybrid mimo transceivers,''
\newblock {\em IEEE Transactions on Wireless Communications}, vol. 18, no. 3,
  pp. 1518--1533, 2019.

\bibitem{yang20}
Y.~Yang, S.~Dang, M.~Wen, S.~Mumtaz, and M.~Guizani,
\newblock ``Bayesian beamforming for mobile millimeter wave channel tracking in
  the presence of doa uncertainty,''
\newblock {\em IEEE Transactions on Communications}, vol. 68, no. 12, pp.
  7547--7562, 2020.

\bibitem{khordad23}
E.~Khordad, I.~B. Collings, S.~V. Hanly, and G.~Caire,
\newblock ``Compressive sensing-based beam alignment schemes for time-varying
  millimeter-wave channels,''
\newblock {\em IEEE Transactions on Wireless Communications}, vol. 22, no. 3,
  pp. 1604--1617, 2023.

\bibitem{pote23dissert}
R.~R. Pote,
\newblock ``Efficient techniques for millimeter wave sensing and beam
  alignment, sparse recovery, and doa estimation,''
\newblock in {\em Ph.D. dissertation, University of California, San Diego},
  2023.

\bibitem{stoica89}
P.~Stoica and Arye Nehorai,
\newblock ``Music, maximum likelihood, and cramer-rao bound,''
\newblock {\em IEEE Transactions on Acoustics, Speech, and Signal Processing},
  vol. 37, no. 5, pp. 720--741, 1989.

\bibitem{wipf04}
D.~{Wipf} and B.~D. {Rao},
\newblock ``Sparse {B}ayesian learning for basis selection,''
\newblock {\em IEEE Trans. on Signal Proc.}, vol. 52, no. 8, pp. 2153--2164,
  Aug 2004.

\bibitem{pote23}
R.~R. Pote and B.~D. Rao,
\newblock ``Light-weight sequential sbl algorithm: An alternative to omp,''
\newblock in {\em IEEE International Conference on Acoustics, Speech and Signal
  Processing (ICASSP)}, 2023, pp. 1--5.

\bibitem{tipping03}
M.~E. Tipping and A.~C. Faul,
\newblock ``Fast marginal likelihood maximisation for sparse {B}ayesian
  models,''
\newblock in {\em Proceedings of the Ninth International Workshop on AI and
  Statistics}. 2003, vol.~R4, pp. 276--283, PMLR.

\bibitem{giri16}
R.~Giri and B.~Rao,
\newblock ``Type i and type ii bayesian methods for sparse signal recovery
  using scale mixtures,''
\newblock {\em IEEE Transactions on Signal Processing}, vol. 64, no. 13, pp.
  3418--3428, 2016.

\end{thebibliography}
\end{document}